\documentclass[envcountsame, envcountsect, runningheads, a4paper]{llncs}
\usepackage[dvipsnames]{xcolor}
\usepackage{amsmath,amsfonts,amssymb}
\usepackage[utf8]{inputenc}
\usepackage{graphicx}
\usepackage{tikz}
\usepackage{hyperref}
\usepackage[backgroundcolor=magenta!10!white]{todonotes}
\usepackage{setspace}
\usepackage{microtype}
\usepackage{enumitem}
\usepackage{ifthen}
\usepackage{caption}
\usepackage{microtype}
\usepackage[space]{cite}
\usepackage{numprint}
\usepackage{cleveref}
\usepackage{multirow}
\usepackage{subcaption}
\usepackage{placeins}

\usepackage{tikz}
\usetikzlibrary{positioning,arrows,automata,calc,shapes}

\definecolor{darkgreen}{rgb}{0.0, 0.5, 0.0}
\definecolor{darkred}{rgb}{0.9, 0.0, 0.0}
\definecolor{darkblue}{rgb}{0.0, 0.0, 0.9}

\tikzset{
	state/.style={rounded rectangle,draw=black,inner sep=1.5mm,minimum width=9mm,minimum height=5mm},
	rstate/.style={rectangle,draw=black,inner sep=1ex,minimum width=9mm,minimum height=5mm},
	istate/.style={minimum width=5mm, minimum height=6mm},
	bullet/.style={circle,draw=black,fill=black,inner sep=0cm, minimum size=0.7mm},
	initial text={},
	every initial by arrow/.style={->,>=stealth'},
	ptran/.style={rounded corners, ->,>=stealth',auto},
	ntran/.style={rounded corners, -,auto}
}

\spnewtheorem{notation}[theorem]{Notation}{\bfseries}{\upshape}
\spnewtheorem{setting}[theorem]{Setting}{\bfseries}{\upshape}

\usepackage{thmtools}

\usetikzlibrary{arrows,
	automata,
	calc,
	decorations,
	decorations.pathreplacing,
	positioning,
	shapes}

\setenumerate[1]{label={(\arabic*)}} 
\numberwithin{equation}{section}
\npdecimalsign{.}
\npthousandsep{,}

\newcommand{\R}{\mathbb{R}}
\renewcommand{\P}{\mathcal{P}}
\newcommand{\prb}{\mathbf{Pr}}
\newcommand{\Ab}{\mathbf{A}}
\newcommand{\Ib}{\mathbf{I}}
\newcommand{\Pb}{\mathbf{P}}
\newcommand{\xb}{\mathbf{x}}
\newcommand{\yb}{\mathbf{y}}
\newcommand{\bb}{\mathbf{b}}
\newcommand{\hb}{\mathbf{h}}
\newcommand{\M}{\mathcal{M}}
\newcommand{\cb}{\mathbf{c}}

\newcommand{\vb}{\mathbf{v}}
\newcommand{\gb}{\mathbf{g}}
\newcommand{\pb}{\mathbf{p}}
\newcommand{\qb}{\mathbf{q}}
\newcommand{\Qb}{\mathbf{Q}}
\newcommand{\Bb}{\mathbf{B}}
\newcommand{\Cb}{\mathbf{C}}

\newcommand{\Jb}{\mathbf{J}}
\newcommand{\db}{\boldsymbol{\delta}}
\newcommand{\zb}{\mathbf{z}}
\newcommand{\all}{{\operatorname{all}}}
\newcommand{\Act}{\operatorname{Act}}
\newcommand{\supp}{\operatorname{supp}}
\newcommand{\last}{\operatorname{last}}
\newcommand{\Dist}{\operatorname{Dist}}
\newcommand{\goal}{\operatorname{goal}}
\newcommand{\fail}{\operatorname{fail}}

\newcommand{\Cyl}{\operatorname{Cyl}}
\renewcommand{\S}{\mathfrak{S}}
\renewcommand{\Pr}{\mathrm{Pr}}
\newcommand{\fin}{{\operatorname{fin}}}
\newcommand{\Paths}{\operatorname{Paths}}
\newcommand{\sucs}{\operatorname{Post}}
\newcommand{\dipro}{\textsc{Dipro}}
\newcommand{\comics}{\textsc{Comics}}
\newcommand{\prism}{\textsc{Prism}}
\newcommand{\ind}{\boldsymbol{\sigma}}

\newif\iflongversion
\longversiontrue
\iflongversion
\newcommand{\citeapp}[1]{\Cref{#1}}
\else
\newcommand{\citeapp}[1]{the full version~\cite{FunkeJB19}}
\fi

\iflongversion
\newcommand{\referapp}{the appendix}
\else
\newcommand{\referapp}{the full version~\cite{FunkeJB19}}
\fi

\begin{document}

\title{Farkas certificates and minimal witnesses for probabilistic reachability constraints}
\author{Florian Funke\orcidID{0000-0001-7301-1550} \and Simon Jantsch\orcidID{0000-0003-1692-2408} \and Christel Baier\orcidID{0000-0002-5321-9343}}
\authorrunning{Florian Funke, Simon Jantsch, Christel Baier}
\titlerunning{Farkas certificates and minimal witnesses}
\institute{Technische Universität Dresden\thanks{This work was funded by DFG grant 389792660 as part of \href{https://perspicuous-computing.science}{TRR~248}, the Cluster of Excellence EXC 2050/1 (CeTI, project ID 390696704, as part of Germany’s Excellence Strategy), DFG-projects BA-1679/11-1 and BA-1679/12-1, and the Research Training Group QuantLA (GRK 1763).}\\
\email{\{florian.funke, simon.jantsch, christel.baier\}@tu-dresden.de}}
% (see \url{https://perspicuous-computing.science})

\maketitle

\begin{abstract}
	This paper introduces \emph{Farkas certificates} for lower and upper bounds on minimal and maximal reachability probabilities in Markov decision processes (MDP), which we derive using an MDP-variant of Farkas' Lemma.
  The set of all such certificates is shown to form a polytope whose points correspond to witnessing subsystems of the model and the property.
  Using this correspondence we can translate the problem of finding minimal witnesses to the problem of finding vertices with a maximal number of zeros.
  While computing such vertices is computationally hard in general, we derive new heuristics from our formulations that exhibit competitive performance compared to state-of-the-art techniques.
	As an argument that asymptotically better algorithms cannot be hoped for, we show that the decision version of finding minimal witnesses is $\operatorname{NP}$-complete even for acyclic Markov chains.
	%\keywords{Certification \and Farkas' Lemma \and Markov decision process \and Reachability problem \and Minimal witness \and Farkas certificate  \and Polytope}
\end{abstract}

\section{Introduction}

\label{sec:intro}

The goal of program verification is to consolidate the user's trust that a given system works as intended, and if this is not the case, to provide her with useful diagnostic information.
Verification tools may, however, contain bugs and so a last grain of insecurity regarding their results always remains.
A widely acknowledged approach to overcome this dilemma has been made in the form of \emph{certifying algorithms}~\cite{BlumK95,McConnellMNS11}. These algorithms provide every result with an accompanying \emph{certificate}, i.e., a token that can be used to verify the result independently and with little ressources. In this way, certificates enable the user (or a third party) to quickly give a mathematically rigorous proof for the correctness of the result \emph{irrespective} of whether the algorithm itself works correctly.

\emph{Counterexamples}, i.e. certificates for the violation of a property, can often be obtained as a byproduct of verification procedures.
What constitutes a counterexample is highly context-dependent.
Finite executions suffice as counterexamples for safety properties and single, possibly infinite, executions are viable counterexamples for LTL~\cite{ClarkeV03}.
Tree-like counterexamples have been considered for fragments of CTL~\cite{ClarkeJLV02}. For a probabilistic system $\M$ and a linear time property $\phi$, the most prominent notion of counterexample to $\Pr_{\M}(\phi)< \lambda$ is a set of paths satisfying $\phi$ whose probability mass is at least $\lambda$ (see~\cite{AbrahamBDJKW14} for a survey). 

Another notion of counterexample for probabilistic systems $\M$ and properties of the form $\Pr_{\M}(\phi)<  \lambda$ are \emph{critical subsystems}~\cite{AbrahamBDJKW14}.
We adopt the reverse perspective and call a subsystem $\M'$ of $\M$ a \emph{witnessing subsystem} for the property $\Pr_{\M}(\phi)\geq \lambda$ if $\Pr_{\M'}(\phi)\geq \lambda$. Small witnessing subsystems offer an insight into what parts of the system are responsible for the satisfaction of the property. Nonetheless, witnessing subsystems can hardly be regarded as viable certificates since verifying $\Pr_{\M'}(\phi)\geq \lambda$ is as hard as checking $\Pr_{\M}(\phi)\geq \lambda$ itself.

In this paper we build a solid bridge between certificates and witnessing subsystems. The systems we consider are modeled as Markov decision processes (MDP), which contain an absorbing goal state representing a desirable outcome. This approach is motivated by the fact that numerous model checking tasks can be reduced to reachability problems \cite{HartSP1983, Vardi1985, CourcoubetisY1995, CourcoubetisY1988, deAlfaro97, VardiW86}.

Using Farkas' Lemma, we introduce certificates for bounds on the minimal and maximal probability to reach the goal state.
We show that the set of these certificates forms a polytope and we provide a direct translation of a certificate to a witnessing subsystems for lower bounded threshold properties.
Thereby, we bridge the gap between an abstract gadget, serving solely as a proof that the result is correct, and a concrete object, containing crucial diagnostic information about \emph{why} the result holds.
Moreover, our translation reduces the computation of minimal witnessing subsystems to a purely geometric problem, for which we provide and evaluate new exact and heuristic algorithms.

\iflongversion
All omitted proofs can be found in the appendix.
\else
All omitted proofs can be found in the full version of this paper~\cite{FunkeJB19}.
\fi

\subsubsection{Contributions.}
\begin{itemize}
	\item Following the concept of certificates in certifying algorithms, we introduce \emph{Farkas certificates} for reachability problems in MDPs (\Cref{tab:overview}).
	\item We give a uniform notion of \emph{witnessing subsystem} (WS) for $\prb^{\max}_{s_0}(\lozenge\goal) \geq \lambda$ and $\prb^{\min}_{s_0}(\lozenge\goal) \geq \lambda$ (\Cref{def:subsystem}). To the best of our knowledge, witnesses for $\prb^{\min}_{s_0}(\lozenge\goal) \geq \lambda$ have not been considered previously.
	\item We establish NP-completeness for finding minimal WS even for acyclic discrete time Markov chains (DTMC) (\Cref{thm:np-completeness}).
	\item Our main result establishes a strong connection between the polytopes of Farkas certificates for $\prb^{\min}_{s_0}(\lozenge\goal) \geq \lambda$ and $\prb^{\max}_{s_0}(\lozenge\goal) \geq \lambda$ and WS of the same property (\Cref{thm:MCS-polytope}).
	In particular, one can read off a minimal WS from a vertex of the polytope with a maximal number of zeros (\Cref{cor:main corollary}).
\item From our polytope characterizations we derive two algorithms for computing minimal WS: one based on vertex enumeration and one based on mixed integer linear programming (\Cref{sec:computing}).
  We also introduce a linear programming based heuristic aimed at computing small WS.
  We evaluate our approach on DTMC and MDP benchmarks, where particularly our heuristics show competitive results compared to state-of-the-art techniques (\Cref{sec:case studies}).
\end{itemize}
\renewcommand{\arraystretch}{1.4}
\renewcommand{\tabcolsep}{4mm}

\setlength{\textfloatsep}{0.5cm}
\begin{table}[tbp]
	\centering
	\captionsetup{width=.75\linewidth, labelfont={bf, up}}
  \begin{tabular}{ |c|c|c|}
		\hline
		Property& Certificate dimension&  Certificate condition  \\ \hline
		$\prb^{\min}_{s_0}(\lozenge\goal) \gtrsim \lambda$ & $\zb\in \R^{S}$ & $\Ab \zb \leq \bb \land \zb(s_0) \gtrsim \lambda$  \\\hline
		$\prb^{\max}_{s_0}(\lozenge\goal) \gtrsim \lambda$ &  $\yb \in \R^{\M}_{\geq 0}$ & $\yb\Ab \leq\delta_{s_0}\land \yb\bb \gtrsim \lambda$   \\\hline
		$\prb^{\min}_{s_0}(\lozenge\goal)\lesssim \lambda$ & $\yb \in \R^{\M}_{\geq 0}$ & $\yb\Ab \geq\delta_{s_0}\land \yb\bb \lesssim \lambda$   \\\hline
		$\prb^{\max}_{s_0}(\lozenge\goal) \lesssim \lambda$ & $\zb\in \R^{S}$ & $\Ab \zb \geq \bb \land \zb(s_0) \lesssim \lambda$  \\\hline
	\end{tabular}
	 \caption{Overview of Farkas certificates for reachability properties in MDPs (where $\lesssim \; \in \{\leq, <\}$ and $\gtrsim \; \in\{\geq, > \}$).}\label{tab:overview}
\end{table}
\subsubsection{Related work.}
The fundament of certifying algorithms has been surveyed in \cite{McConnellMNS11}.
In the context of model checking, the most prominent approach for the certification of a positive result has been to construct a proof of the property in the system~\cite{Namjoshi01,PeledPZ01,BernasconiMSZG17}.
Rank-based certificates for the emptiness of a certain automaton ~\cite{KupfermanV04} can be used to certify positive model checking results.
Model checking MDPs in the presence of multiple objectives has been studied in~\cite{EtessamiKVY08, ForejtKNPQ11}.

Heuristic approaches for computing small witnessing subsystems in DTMCs have been proposed in~\cite{AljazzarL06,AljazzarL10,JansenAKWKB11, JansenAZWSKB12,JansenWAZKBS14} and implemented in the tool \comics{} \cite{JansenAVWKB12}.
Witnessing subsystems in MDPs have been considered in~\cite{AljazzarL09, AndresDR08} and~\cite{BrazdilCCFK15}, which focuses on succinctly representing witnessing schedulers.
The mixed integer linear programming (MILP) formulation of~\cite{WimmerJABK12,WimmerJAKB14} allows for an exact computation of minimal witnessing subsystems for the property $\prb^{\max}_{s_0}(\lozenge\goal) \gtrsim \lambda$.
NP-completeness of computing minimal witnessing subsystems in MDPs was shown in \cite{ChadhaV10}, but the exact complexity has, to the best of our knowledge, not been determined for DTMCs (the problem was conjectured to be NP-complete in~\cite{WimmerJABK12}).

Minimal probabilistic counterexamples given as sets of paths can be computed by reframing the problem as a $k$-shortest-path problem~\cite{HanKD09,HanK07}.
Regular expressions have been considered to succinctly represent the set of paths in~\cite{DammanHK08}, and extensions were proposed in~\cite{WimmerBB09,BraitlingWBJA11}.
The tool \dipro{} \cite{AljazzarLLS11} computes probabilistic counterexamples, and a translation of these to fault trees was given in \cite{KuntzLL2011}.
Another, learning-based, approach~\cite{BrazdilCCFKKPU2014} also enumerates paths and produces a witnessing subsystem as a byproduct.
But none of these approaches considers state-based minimality.
Probabilistic counterexamples can be used to automatically guide iterative and refinement-based model checking techniques~\cite{ClarkeGJLV03, ChadhaV10,ChatterjeeCD14,CeskaHJK19,CeskaJJK19,HermannsWZ08}.

Farkas' Lemma is a well-known source of certificates for the (in)feasibility of tasks in combinatorial optimization, operations research, and economics, as presented in the detailed historical account given in \cite[pp. 209--226]{Schrijver1986} as well as \cite[Chapter 2]{Mangasarian94} and \cite{ColonSS03, Vohra2006, NaimanS2017}. The lecture notes \cite{Schrijver2017} contain a rich variety of applications of linear programming in general and Farkas' Lemma in particular.

\section{Preliminaries}
\label{sec:prelim}

\subsubsection{Polyhedra and Farkas' Lemma.}
Throughout the article we write the dot product of two vectors $\xb, \yb\in \R^n$ as $\xb\yb$ or $\xb\cdot \yb$.
A \emph{halfspace} in $\R^n$ is a set $H = \{\vb\in \R^n\mid \mathbf{a}\cdot \vb\leq b \}$ for some non-trivial $\mathbf{a}\in\R^n$ and $b\in\R$.
A \emph{polyhedron} is the intersection of finitely many halfspaces, and a \emph{polytope} is a bounded polyhedron. A \emph{face} of a polyhedron $P$ is a subset $F\subseteq P$ of the form $F= \{\xb\in P\mid \mathbf{a}\cdot \xb=\max\{ \mathbf{a}\cdot\yb \mid \yb\in P\}\}$ for some $\mathbf{a}\in \R^n$. A \emph{vertex} of $P$ is a face consisting of only one point.

Farkas' Lemma \cite{Farkas1902} is part of the fundament of polyhedra theory and linear programming. It provides a natural source of certificates showing the infeasibility of a given system of inequalites, or in other words, the emptiness of the polyhedron described by the system. We will use it in the following version.

\begin{lemma}[Farkas' Lemma, cf. {\cite[Corollary 7.1f on p. 90]{Schrijver1986}}]
	\label{lem:Farkas}
	Let $\Ab\in \R^{m\times n}$ and $\bb\in \R^m$. Then there exists $\zb\in \R^n_{\geq 0}$ with $\Ab\zb \leq \bb$ if and only if there does \emph{not} exist $\yb\in \R^m_{\geq 0}$ with $\yb\Ab \geq 0 \land \yb\bb< 0$.
	\iffalse
	\begin{alignat*}{3}
	 	&\text{(1)} \quad\exists\,\zb\in \R^n_{\geq 0}. \;\; & \Ab\zb = \bb \;\;\iff\;\; &\neg\,\exists\, \yb\in \R^m.\;\; \yb\Ab \geq 0 \land \yb\bb< 0\\
		&\text{(2)} \quad\exists\,\zb\in \R^n_{\geq 0}. \;\; & \Ab\zb \leq \bb \;\;\iff\;\; &\neg\,\exists\, \yb\in \R^m_{\geq 0}.\;\; \yb\Ab \geq 0 \land \yb\bb< 0\\
		&\text{(3)} \quad\exists\,\zb\in \R^n. \;\; & \Ab\zb \leq \bb \;\;\iff\;\; &\neg\,\exists\, \yb\in \R^m_{\geq 0}.\;\; \yb\Ab = 0 \land \yb\bb< 0
	\end{alignat*}
	\fi
\end{lemma}

\subsubsection{Markov decision processes.}
A \emph{Markov decision process} (MDP) is a tuple $\M = (S, \Act, \iota, \Pb)$, where $S$ is a finite set of \emph{states}, $\Act$ is a finite set of \emph{actions}, $\iota$ is a probability distribution on $S$ called the \emph{initial distribution} of $M$, and $\Pb\colon S\times\Act\times S\to [0,1]$ is the \emph{transition probability function} where we require $\sum_{s'\in S} \Pb(s, \alpha, s') \in\{0,1\}$ for all $s\in S$ and $\alpha\in\Act$. An action $\alpha$ is \emph{enabled} in state $s\in S$ if $\sum_{s'\in S} \Pb(s, \alpha, s') =1$. The set of enabled actions at state $s$ are denoted by $\Act(s)$, and we require $\Act(s) \neq \varnothing$ for all $s\in S$. A \emph{path} in an MDP $\M$ is an infinite sequence $s_0\alpha_0s_1\alpha_1...$ such that $\Pb(s_i,\alpha_i,s_{i+1})>0$ for all $i\geq 0$. A finite path is a finite sequence $\pi = s_0\alpha_0s_1\alpha_1... s_n$ with the same condition for all $0\leq i\leq n-1$. In this case, we define $\last(\pi) =s_n$. Denote by $\Paths(\M)$ and $\Paths_\fin(\M)$ the set of infinite and finite paths in $\M$.

A \emph{discrete-time Markov chain} (DTMC) is an MDP with a single action which is enabled at every state. If $\M$ is a DTMC, then $\Paths(\M)$ carries a probability measure, where the associated $\sigma$-algebra is generated by the cylinder sets $\Cyl(\tau) = \{ \pi\in\Paths(\M)\mid \pi\text{ has prefix } \tau\}$ of finite paths $\tau = s_0s_1...s_n$ in $\M$ with probability $\Pr(\Cyl(\tau)) = \iota(s_0)\cdot \prod_{0 \leq i < n} \Pb(s_i,s_{i+1})$ (fore more details see \cite[Section 10.1]{BaierK2008}).
In the following we denote for a finite set $X$ the set of probability distributions on $X$ by $\Dist(X)$. Given $\mu\in\Dist(X)$ let the \emph{support} of $\mu$ be $\supp(\mu) = \{x\in X\mid \mu(x) >0 \}$.

A \emph{deterministic scheduler} is a function $\S\colon \Paths_\fin(\M) \to \Act$ such that $\S(\pi)\in\Act(\last(\pi))$ and a \emph{randomized scheduler} is a function $\S\colon \Paths_\fin(\M) \to \Dist(\Act)$ such that $\supp(\S(\pi))\subseteq \Act(\last(\pi))$ for all $\pi\in\Paths_\fin(\M)$. 
Given a deterministic (or randomized) scheduler $\S$, a path $\pi = s_0\alpha_0s_1\alpha_1...$ in $\M$ is an \emph{$\S$-path} if $\alpha_i = \S(s_0\alpha_0...s_i)$ (or $\alpha_i\in \supp(\S(s_0\alpha_0...s_i))$) for all $i\geq 0$.

 We denote by $\Pr^\S$ the probability measure on infinite $\S$-paths (see \cite[Definition 10.92 on page 843]{BaierK2008} for more details). If we replace $\iota$ with the distribution concentrated on state $s$, then we obtain a probability measure $\Pr^\S_{\M, s}$ or short $\Pr^\S_{s}$ on infinite $\S$-paths starting in $s$.
The scheduler is \emph{memoryless} if $\S(\pi) = \S(\last(\pi))$ for all $\pi\in\Paths_\fin(\M)$. We abbreviate memoryless deterministic schedulers as \emph{MD-schedulers} and memoryless randomized schedulers as \emph{MR-schedulers}.

Given a state $t\in S$, we let 
\[\prb_s^{\max}(\lozenge t)= \sup_{\S} \; \Pr_s^\S(\lozenge t) \quad \text{and} \quad \prb_s^{\min}(\lozenge t)= \inf_{\S} \; \Pr_s^\S(\lozenge t) \]
denote the maximal and minimal probability to reach $t$ eventually when starting in $s$
and set $\prb^{\min}(\lozenge t) = (\prb_s^{\min}(\lozenge t))_{s \in S}$ and $\prb^{\max}(\lozenge t) = (\prb_s^{\max}(\lozenge t))_{s \in S}$.
The supremum and infimum is indeed attained by an MD-scheduler \cite[Lemmata 10.102 and 10.113]{BaierK2008}, thus justifying the superscripts.
\begin{setting}
	\label{setting:mdp}
	Henceforth we will assume that $\M = (S_\all, \Act, \iota, \Pb)$ has a unique initial state $s_0 \in S$ and two distinguished absorbing states $\fail$ and $\goal\in S_\all$, i.e., $\Pb(\goal, \alpha, s) = 0$ for all $\alpha\in\Act$ and $s\in S_\all$ with $s\neq\goal$, and likewise for $\fail$.
  Here $\goal$ represents a desirable outcome of the modeled system and $\fail$ an  outcome that is to be avoided. We use the notation $S = S_\all \setminus \{\fail, \goal\}$, we assume that every state $s\in S$ is reachable from $s_0$. We also assume that under every scheduler $\fail$ or $\goal$ is reachable from any state, i.e., $\prb^{\min}_s(\lozenge (\goal\lor \fail)) > 0$ for all $s\in S$. If $\M$ does not satisfy this condition from the start, we can apply a standard preprocessing step, which is essentially given by taking the MEC quotient of $\M$, see \cite{deAlfaro97b, deAlfaro97} and also \cite{Ciesinski08}. While it is often easier to verify the condition $\prb^{\min}_s(\lozenge (\goal\lor\fail)) > 0$, it is in fact equivalent to $\prb^{\min}_s(\lozenge (\goal\lor\fail)) = 1$ (see \citeapp{lem: almost sure BSCC}).
	
	Whenever suitable, we denote by $\M$ also the set of enabled state-action pairs, i.e., $\M =\{(s,\alpha)\in S\times\Act\mid \alpha\in\Act(s)\}$. Let $\Ab \in \R^{\M \times S}$ be defined by
	\[\mathbf{A}((s,\alpha),t) =
	\begin{cases}
	1 - \Pb(s,\alpha,s), & \text{if } s = t \\
	-\Pb(s,\alpha,t), & \text{if } s\neq t
	\end{cases}\]
  We denote by $\bb = (\bb(s,\alpha))_{(s,\alpha)\in\M}\in \R^{\M}$ with $\bb(s, \alpha) = \Pb(s, \alpha, \goal)$ and by $\delta_{s_0}$ the probability distribution that assigns $1$ to $s_0$, and $0$ to all other states.
\end{setting}

The vectors $\prb^{\min}(\lozenge\goal)$ and $\prb^{\max}(\lozenge\goal)$ can be characterized using the following linear programs. Although this characterization is well-known, we give a proof in \referapp{} due to slight differences with the standard literature.

\begin{restatable}[LP characterization, cf. {\cite[Lemma 8]{BiancoA95}}]{proposition}{prminmaxlp}
  \label{prop:LP version of pr_min}
	Let $\M$ be an MDP as in Setting \ref{setting:mdp} and let $\db \in \R^n_{>0}$. Then the vectors $\prb^{\min}(\lozenge\goal)$ and $\prb^{\max}(\lozenge\goal)$ are, respectively, the \emph{unique} solution of the LPs
	\begin{equation*} \label{eq: LP for min}
	\max \, \db \cdot \zb \;\text{ s.t. } \; \Ab \zb \leq \bb \quad\;\; \text{ and } \quad\;\; \min \, \db \cdot \zb \;\text{ s.t. } \; \Ab \zb \geq \bb.
	\end{equation*}
\end{restatable}

\section{Farkas certificates for reachability in MDPs}
\allowdisplaybreaks
\label{sec:certificates}

In this section we establish certificates for the following statements:
\begin{enumerate}[leftmargin =10mm]
	\item All schedulers $\S$ satisfy $\Pr_{s_0}^\S(\lozenge\goal) \gtrsim \lambda$ (i.e., $\prb_{s_0}^{\min}(\lozenge\goal) \gtrsim \lambda$).
	\item Some scheduler $\S$ satisfies $\Pr_{s_0}^\S(\lozenge\goal) \gtrsim \lambda$ (i.e., $\prb_{s_0}^{\max}(\lozenge\goal) \gtrsim \lambda$).
	\item All schedulers $\S$ satisfy $\Pr_{s_0}^\S(\lozenge\goal) \lesssim \lambda$ (i.e., $\prb_{s_0}^{\max}(\lozenge\goal) \lesssim \lambda$).
	\item Some scheduler $\S$ satisfies $\Pr_{s_0}^\S(\lozenge\goal) \lesssim \lambda$ (i.e., $\prb_{s_0}^{\min}(\lozenge\goal) \lesssim \lambda$).
\end{enumerate}
where $\lesssim \; \in \{\leq, <\}$ and $\gtrsim \; \in\{\geq, >\}$.
The basis of our construction is the LP characterization of the probabilities above and, crucially, Farkas' Lemma.

\subsubsection{Certificates for universally-quantified statements.}

In order to deal with the cases (1) and (3), we need the following lemma proved in \referapp{}.

\begin{restatable}{lemma}{monotonicityprminprmax}
	\label{lem:monotonicity}
	For $\Ab \in \R^{\M \times S}, \bb\in \R^{\M}$ as in Setting \ref{setting:mdp}, we have for all $\zb\in \R^{S}:$
\begin{gather*}
	\Ab\zb \leq \bb \implies \zb \leq \prb^{\min}(\lozenge\goal)\\
	\Ab\zb \geq \bb \implies \zb \geq \prb^{\max}(\lozenge\goal)
\end{gather*}
\end{restatable}

\begin{corollary}\label{cor:certificates for all quantified}
	For  $\gtrsim \; \in \{\geq , >\}$ and  $\lesssim \; \in \{\leq , <\}$ we have
\begin{gather*} 
	\prb^{\min}_{s_0}(\lozenge\goal) \gtrsim \lambda \iff \exists \zb\in\R^S.\; \Ab \zb \leq \bb \land \zb(s_0) \gtrsim  \lambda\\
	\prb^{\max}_{s_0}(\lozenge\goal) \lesssim \lambda \iff \exists \zb\in\R^S.\; \Ab \zb \geq \bb \land \zb(s_0) \lesssim \lambda
\end{gather*}
\end{corollary}
\begin{proof}
  For the direction from left to right, we take $\zb$ to be $\prb^{\min}(\lozenge\goal)$. The opposite direction follows from \Cref{lem:monotonicity}.
\qed\end{proof}

The right hand sides of \Cref{cor:certificates for all quantified} provide \emph{certifying} formulations for problems (1) and (3): to check whether the corresponding threshold statement holds, one must merely check whether $\zb$ satisfies the inequalities, rather than checking whether $\prb^{\min / \max}_{s_0}(\lozenge\goal)$ was computed correctly.
If the threshold condition is satisfied, then the vectors $\prb_{s_0}^{\min/\max}(\lozenge \goal)$ are also valid certificates.

\subsubsection{Certificates for existentially-quantified statements.}\label{sec:certificates for E}

To find certificates for the cases (2) and (4), we calculate:
\begin{align*}
  &\quad\prb^{\min}_{s_0}(\lozenge\goal) < \lambda \\
\overset{\text{Cor. \ref{cor:certificates for all quantified}}}{\iff} &\quad\neg \,\exists \zb\in\R^S_{\geq 0}. \, \,\;  \Ab\zb \leq  \bb \land \zb(s_0) \geq \lambda\\
\iff  &\quad\neg \,\exists \zb\in\R^S_{\geq 0}. \, \,\;  
\begin{pmatrix}
\Ab \\
-1 \:0 \ldots 0
\end{pmatrix}\zb \leq  
\begin{pmatrix}
\bb\\
-\lambda
\end{pmatrix}\\
\overset{\text{Lem. \ref{lem:Farkas}}}{\iff}  &\quad\exists \yb \in \R^{\M}_{\geq 0}, y^*\geq 0.\, \,\; (\yb, y^*)
\begin{pmatrix}
\Ab \\
-1 \:0 \ldots 0
\end{pmatrix}\geq 0 \land
(\yb, y^*)
\begin{pmatrix}
\bb\\
-\lambda
\end{pmatrix} < 0\\
\iff &\quad\exists \yb \in \R^{\M}_{\geq 0}.\, \,\; \yb\Ab \geq\delta_{s_0}\land \yb\bb < \lambda.
\end{align*}
For non-strict inequalities, we apply Farkas' Lemma in the opposite direction:
\begin{align*}
   &\quad\prb^{\min}_{s_0}(\lozenge\goal) \leq \lambda \\
\overset{\text{Cor. \ref{cor:certificates for all quantified}}}{\iff}  &\quad\neg \,\exists \zb\in\R^S_{\geq 0}. \, \,\;  \Ab\zb \leq  \bb \land \zb(s_0) > \lambda\\
  \iff &\quad\neg \,\exists \zb\in\R^S_{\geq 0}, z^* \geq 0.\, \,\;  \begin{pmatrix}
   - \Ab & \bb
    \end{pmatrix}
    \begin{pmatrix}
 	 \zb \\
 	 z^*
 	 \end{pmatrix} \geq 0 \land
  \begin{pmatrix}
    -\delta_{s_0} &   \lambda
  \end{pmatrix} 
   \begin{pmatrix}
  \zb \\
  z^*
  \end{pmatrix} < 0  
  \\
 \overset{\text{Lem. \ref{lem:Farkas}}}{\iff}&\quad\exists \yb \in \R^{\M}_{\geq 0}.\, \,\; \yb \begin{pmatrix}
    -\Ab & \bb 
    \end{pmatrix} \leq \begin{pmatrix}
    -\delta_{s_0} & \lambda
  \end{pmatrix} \\
    \iff &\quad\exists \yb \in \R^{\M}_{\geq 0}.\, \,\; \yb\Ab \geq\delta_{s_0}\land \yb\bb \leq \lambda.\label{eq:Farkas mdp}
\end{align*}
The deductions for $\prb^{\max}(\lozenge\goal)$ are analogous, so that we get: 
\begin{proposition}\label{prop:certificates for ex quantified}
	For  $\gtrsim \; \in \{\geq , >\}$ and  $\lesssim \; \in \{\leq , <\}$ we have
	\begin{gather*} 
	\prb^{\min}_{s_0}(\lozenge\goal) \lesssim \lambda \iff \exists \yb\in\R^{\M}_{\geq 0}.\, \,\; \yb\Ab \geq \delta_{s_0} \land \yb\bb \lesssim  \lambda\\
	\prb^{\max}_{s_0}(\lozenge\goal) \gtrsim \lambda \iff \exists \yb\in\R^{\M}_{\geq 0}.\, \,\;\yb\Ab \leq \delta_{s_0} \land \yb\bb \gtrsim  \lambda
	\end{gather*}
\end{proposition}
Together, \Cref{cor:certificates for all quantified} and \Cref{prop:certificates for ex quantified} give us all certificate conditions of~\Cref{tab:overview}.
\section{Minimal witnesses for reachability in MDPs}
\label{sec:mws}

In this section we consider the following problem: Given an MDP $\M$ that satisfies the property $\prb_{\M, s_0}^{\min}(\lozenge\goal)\geq \lambda$ (or $\prb_{\M, s_0}^{\max}(\lozenge\goal)\geq \lambda$), find a small subsystem $\M'$ of $\M$ that still satisfies these thresholds.
Such a subsystem is a witness to the satisfaction of the property in $\M$.
We first define subsystems and consider different measures of size which we show to be equivalent.
Then we deal with the question of finding minimal witnessing subsystems.

\subsubsection{Subsystems, witnesses and notions of minimality.}

Our definition of subsystem is essentially the same to the definition in~\cite{WimmerJABK12,WimmerJAKB14} that was used for witnessing subsystems of $\prb_{\M, s_0}^{\max}(\lozenge\goal)\gtrsim\lambda$. From now on we restrict our attention to properties of the form  $\prb_{\M, s_0}^{\min/\max}(\lozenge\goal)\gtrsim\lambda$.
One can deal with upper bounds by exchanging the roles of $\fail$ and $\goal$ and invoking the equality $\prb_{\M,s_0}^{\min}(\lozenge\goal) = 1 - \prb_{\M,s_0}^{\max}(\lozenge\fail)$, which holds by the conditions of Setting~\ref{setting:mdp}.

Intuitively, a subsystem $\M'$ of $\M$ contains a subset of states of $\M$, and a transition of $\M$ originating in a state of $\M'$ remains unchanged in $\M'$ or is redirected to $\fail$ (instead of explicitely redirecting to $\fail$, sub-stochastic distributions are used in~\cite{WimmerJABK12,WimmerJAKB14} with the same effect).

\begin{definition}[Subsystem and witness]\label{def:subsystem}
	Let $\mathcal{M} = (S_\all,\Act,s_0, \Pb)$ be an MDP as in Setting \ref{setting:mdp}. A \emph{subsystem} $\mathcal{M}' \subseteq\mathcal{M}$ is an MDP $\M' = (S'_\all,\Act,s_0,\Pb')$ with $\fail, \goal \in S'_\all\subseteq S_\all$, $\Act_{\M'}(s) = \Act_{\M}(s)$ for all $s\in S'_\all$, and for all $s,t \in S'_\all$ with $t\neq \fail$ and $\alpha\in\Act$ we have
	\[ \Pb'(s,\alpha, t) >0 \Longrightarrow \Pb'(s,\alpha, t) = \Pb(s,\alpha, t).\]
	We say that the states $S_\all\setminus S'_\all$ and the transitions $(s,\alpha, t)$ with $\Pb(s,\alpha, t) > 0$ and $\Pb'(s,\alpha, t) = 0$ have been \emph{deleted} in $\M'$. A \emph{witness} for $\prb_{\M, s_0}^{\min/\max}(\lozenge\goal)\gtrsim \lambda$ is a subsystem $\M'\subseteq\M$ such that $\prb_{\M', s_0}^{\min/\max}(\lozenge\goal)\gtrsim \lambda$.
\end{definition}
\begin{remark}
	The condition $\Act_{\M'}(s) = \Act_{\M}(s)$ ensures that the probability of a deleted transition $(s,\alpha, t)$ is added to $(s, \alpha, \fail)$. This is essential for witnesses for $\prb_{\M, s_0}^{\min}(\lozenge\goal)\gtrsim \lambda$ as one could otherwise remove entire actions causing low probabilities and obtain greater $\prb^{\min}$ in $\M'$ than in $\M$ as a result. For witnesses of $\prb_{\M', s_0}^{\max}(\lozenge\goal)\gtrsim \lambda$ one could delete this condition, thus leading to the notion of ~\cite{WimmerJABK12,WimmerJAKB14}.
\end{remark}

\begin{example} \Cref{fig:MDP} depicts an MDP and \Cref{fig:subsystem} indicates the subsystem that is obtained by deleting the state $t$ and additionally the transition $(u, \alpha, s_0)$.
	
\begin{figure}[tbp]
	\captionsetup{width=0.75\linewidth, labelfont={bf, up}}
  \begin{subfigure}[b]{0.5\columnwidth}
	\centering
  \scalebox{1.3}{
  		
  		\begin{tikzpicture}[x=18mm,y=15mm,font=\scriptsize]
  		
  		\node[state] (s) {$s_0$};
  		\node[state] (t) [above right = 0.8cm and 1.5cm of s]  {$t$};
  		\node[state] (u) [below right = 0.8cm and 1.5cm of s]  {$u$};
  		\node[state] (v) [below right = 1cm and 1cm of t]  {$v$};
  		\node[state] (f) [right = 1.3cm of u] {$\fail$};
  		\node[state] (g) [above = 2.1cm of f]  {$\goal$};
  		
  		\node[bullet] (sst) [above right = 0.5cm and 0.7cm of s] {};
  		\node[bullet] (stu) [right = 0.5cm of s] {};
  		\node[bullet] (tuv) [below right = 0.5cm and 0.5cm of t] {};
  		\node[bullet] (uvws) [above right = 0.3cm and 0.5cm of u] {};
  		\node[bullet] (vgf) [right = 0.5cm of v] {};
  		\node[bullet] (wgv) [above right = 0.3cm and 0.5cm of v] {};
  		
  		{\color{blue!50!black}
  			\draw [ntran] (s) -- node[above=12pt,pos=.6]{$\alpha$} (sst);
  			\draw (sst) edge[ptran, bend right, in=210] coordinate[pos=.3] (bss) (s);
  			\draw (sst) edge[ptran, bend right,out=40, in=150] coordinate[pos=.3] (bst1) (t);
  			\draw (bss) to[bend left] (bst1);}
  		
  		{\color{green!50!black}
  			\draw [ntran] (s) -- node[above=-0.04,pos=.55]{$\beta$} (stu);
  			\draw (stu) edge[ptran, bend right, in=160] coordinate[pos=.3] (bst) (t);
  			\draw (stu) edge[ptran, bend right,out=0, in=160] coordinate[pos=.3] (bsu) (u);
  			\draw (bst) to[bend left] (bsu);}
  		
  		{\color{blue!50!black}
  			\draw [ntran] (t) -- node[right=0.02,pos=.2]{$\alpha$} (tuv);
  			\draw (tuv) edge[ptran, bend right, in=175] coordinate[pos=.2] (btu) (u);
  			\draw (tuv) edge[ptran, bend right, out = 60, in=175] coordinate[pos=.5] (btv) (v);
  			\draw (btu) to[bend right] (btv);}
  		
  		{\color{blue!50!black}
  			\draw [ntran] (u) -- node[right=7pt,pos=1]{$\alpha$} (uvws);
  			\draw (uvws) edge[ptran, bend left] coordinate[pos=.35] (buv) (v);
  			\draw [ptran] (uvws) -- node[left,pos=.5]{} ($(s) +(2.2cm,-1.7cm)$) -- ($(s) - (0,1.7cm)$) -- (s);
  			\draw (buv) to[bend left] ($(uvws) + (0.2cm,0)$);
  			\draw ($(uvws) + (0.2cm,0)$) to[bend left] ($(uvws) - (0,0.22cm)$);}

  		{\color{olive}
  			\draw [ntran] (v) -- node[above,pos=.8]{$\gamma$} (vgf);
  			\draw (vgf) edge[ptran, bend right, in=175] coordinate[pos=.4] (bvg) (g);
  			\draw (vgf) edge[ptran, bend left, out = 20, in=175] coordinate[pos=.2] (bvf) (f);
  			\draw (bvf) to[bend right] (bvg);}
  		
  		{\color{blue!50!black}
  			\draw [ntran] (v) -- node[left=1pt,pos=.9]{$\alpha$} (wgv);
  			\draw (wgv) edge[ptran, bend right, in=175] coordinate[pos=.4] (bwg) (g);
  			\draw (wgv) edge[ptran, bend left, out = -20, in=200] coordinate[pos=.2] (bwv) (t);
  			\draw (bwg) to[bend right] (bwv);}
  		
  		{\color{olive}\draw [ptran] (s) -- node[above=-.05,pos=.3]{$\gamma$} (u);}
  		{\color{green!50!black}\draw [ptran] (u) -- node[right=-.04,pos=.5]{$\beta$} (t);}
  		\end{tikzpicture}}
	\subcaption{}
  \label{fig:MDP}
  \end{subfigure}
  \begin{subfigure}[b]{0.5\columnwidth}
	\centering
	\scalebox{1.3}{

		\begin{tikzpicture}[x=18mm,y=15mm,font=\scriptsize]
		
		\node[state] (s) {$s_0$};
		\node[state] (u) [below right = 0.8cm and 1.5cm of s]  {$u$};
		\node[state] (v) [below right = 1cm and 1cm of t]  {$v$};
		\node[state] (f) [right = 1.3cm of u] {$\fail$};
		\node[state] (g) [above = 2.1cm of f]  {$\goal$};
		
		\node[bullet] (sst) [above right = 0.5cm and 0.7cm of s] {};
		\node[bullet] (stu) [right = 0.5cm of s] {};
		\node[bullet] (uvws) [above right = 0.3cm and 0.5cm of u] {};
		\node[bullet] (vgf) [right = 0.5cm of v] {};
		\node[bullet] (wgv) [above right = 0.3cm and 0.5cm of v] {};
		
		{\color{blue!50!black}
			\draw [ntran] (s) -- node[above=12pt,pos=.6]{$\alpha$} (sst);
			\draw (sst) edge[ptran, bend right, in=210] coordinate[pos=.3] (bss) (s);
			{\color{red} \draw[thick, dashed,ptran] (sst) -- ($(sst) + (0,1.03cm)$) -- ($(g) + (0.5cm,0.5cm)$) -- ($(f) + (0.5cm,0)$) -- (f);}
			\draw (bss) to[bend left] ($(bst1) - (0.08cm,0)$);}
		
		{\color{green!50!black}
			\draw [ntran] (s) -- node[above=-0.04,pos=.55]{$\beta$} (stu);
			{\color{red} \draw[thick, dashed,ptran] (stu) -- ($(stu) - (0,2cm)$) -- ($(f) - (0,0.69cm)$) -- (f);}
			\draw (stu) edge[ptran, bend right,out=0, in=160] coordinate[pos=.3] (bsu) (u);
			\draw (bsu) to[bend left] ($(stu) - (0,0.4cm)$);}
		
		{\color{blue!50!black}
			\draw [ntran] (u) -- node[right=7pt,pos=1]{$\alpha$} (uvws);
			\draw (uvws) edge[ptran, bend left] coordinate[pos=.35] (buv) (v);
			{\color{red}\draw [ptran,thick, dashed] (uvws) -- ($(uvws) +(0,-0.2cm)$) -- ($(uvws) + (1cm,-0.2cm)$) -- (f);}
			\draw (buv) to[bend left] ($(uvws) + (0.2cm,0)$);
			\draw ($(uvws) + (0.2cm,0)$) to[bend left] ($(uvws) + (0.08cm,-0.18cm)$);}
		
		{\color{olive}
			\draw [ntran] (v) -- node[above,pos=.8]{$\gamma$} (vgf);
			\draw (vgf) edge[ptran, bend right, in=175] coordinate[pos=.4] (bvg) (g);
			\draw (vgf) edge[ptran, bend left, out = 20, in=175] coordinate[pos=.2] (bvf) (f);
			\draw (bvf) to[bend right] (bvg);}
		
		{\color{blue!50!black}
			\draw [ntran] (v) -- node[left=1pt,pos=.9]{$\alpha$} (wgv);
			\draw (wgv) edge[ptran, bend right, in=175] coordinate[pos=.2] (bwg) (g);
			{\color{red}\draw[thick, dashed] (wgv) -- ($(wgv) + (0.73cm,0)$);}
			\draw (bwg) to[bend left] ($(wgv) + (0.2cm,0)$);}
		
		{\color{olive}\draw [ptran] (s) -- node[above=-.05,pos=.3]{$\gamma$} (u);}
		{\color{red}\draw[thick, dashed] [ptran] (u) -- node[below=-0.03,pos=.5]{\color{green!50!black}$\beta$} (f);}
		\end{tikzpicture}
		
	}
	\subcaption{}
  \label{fig:subsystem}
  \end{subfigure}
  \caption{An MDP (with omitted probabilities (\subref{fig:MDP})) and a subsystem (\subref{fig:subsystem}), where redirected transitions are dashed.}
\end{figure}
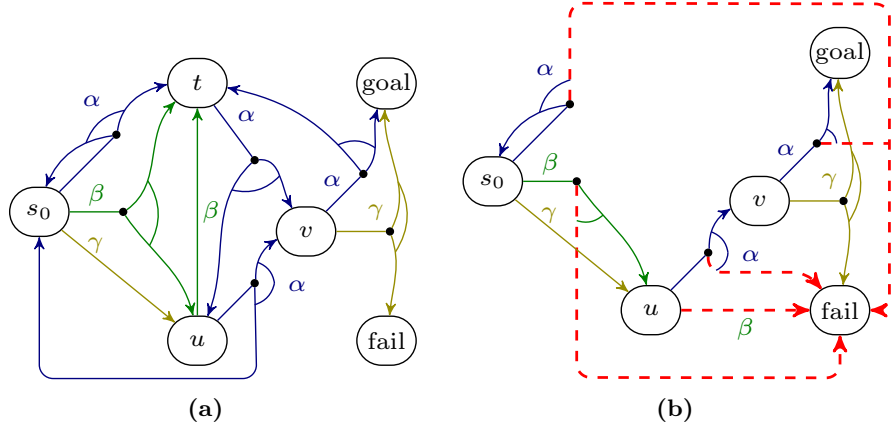
\end{example}
The following lemma ensures that we can use the subsystems as witnesses for both $\prb_{\M, s_0}^{\max}(\lozenge\goal)\gtrsim \lambda$ and $\prb_{\M, s_0}^{\min}(\lozenge\goal)\gtrsim \lambda$.
\begin{restatable}{lemma}{schedulerssubsys}
  Let $\M$ be an MDP as in Setting \ref{setting:mdp} and $\M' \subseteq \M$. Then:
     \[\prb^{\min}_{\M', s_0}(\lozenge\goal) \leq \prb^{\min}_{\M, s_0}(\lozenge\goal) \;\;\text{ and }\;\; \prb^{\max}_{\M', s_0}(\lozenge\goal) \leq \prb^{\max}_{\M, s_0}(\lozenge\goal)\]
 \end{restatable}
\noindent We consider the following notions of minimality for subsystems:
\begin{enumerate}[leftmargin=12mm, rightmargin=8mm]
	\item \emph{State-minimality:} $|S'_\all|$ is minimal.
	\item \emph{Transition-minimality:} The number of transitions, i.e. triples $(s,\alpha,t)$ satisfying $\Pb'(s,\alpha, t) > 0$, is minimal;
	\item \emph{Size-minimality:} The sum of states and transitions is minimal.
\end{enumerate}

Depending on the situation, one notion might be more suitable than the others. However, in \citeapp{lem:minreductions} we show that finding transition-minimal (respectively, size-minimal) witnesses can be reduced to finding state-minimal witnesses with a linear (respectively, quadratic) blow-up. We will therefore restrict ourselves to state-minimality for the rest of this paper.

\subsubsection{NP-completeness of finding minimal witnesses for DTMCs.}
\label{sec:npcompleteness}
In this section we determine the computational complexity of the \emph{witness problem}: Given a DTMC $\M$, a positive integer $k$, and a rational number $\lambda\in[0,1]$, decide whether there exists a witness $\M'\subseteq \M$ for $\Pr_{\M, s_0}(\lozenge\goal)\geq \lambda$ with at most $k$ states. The corresponding problem for MDPs is known to be NP-complete~\cite{ChadhaV10,WimmerJAKB14}\footnote{Although the framework in~\cite{ChadhaV10} considers a richer logic, the hardness proof uses only probabilistic reachability formulas such as the ones we consider.}. In this section we show that the witness problem is already NP-complete for acyclic DTMCs, where acyclicity means that the underlying graph with $V = S$ and $E = \{(s,t)\in S\times S\mid \Pb(s,t) > 0\}$ is acyclic (as before, we take $S = S_\all \setminus \{\goal,\fail\}$). This answers a conjecture of \cite{WimmerJABK12} in the affirmative and also shows NP-completeness of finding minimal witnesses for $\prb_{\M, s_0}^{\min}(\lozenge\goal)\geq \lambda$. 

\begin{restatable}{theorem}{mwanpcompleteness}
\label{thm:np-completeness}
	The witness problem is NP-complete for acyclic DTMCs.
\end{restatable}
\begin{proof}[Sketch]
	An NP-algorithm for the witness problem is given by guessing a set of states of size $k$ and verifying in polynomial time that the corresponding subsystem satisfies $\Pr_{\M', s_0}(\lozenge\goal)\geq \lambda$.
	
	For hardness, we give a reduction from the \emph{clique problem}, which is among Karp's 21 NP-complete problems \cite{Karp1972}.
	The idea is the following: Given an instance of the clique problem with graph $G=(V,E)$ and integer $k$, construct an acyclic Markov chain $\M$ with states $S = \{s_0\}\cup V \cup E \cup\{\goal, \fail\}$ and edges from each vertex $v \in V$ to all edges to which it is incident.
	Then the existence of a $k$-clique can be reduced to the existence of a ``saturated'' subsystem in $\M$ with $k$ states in $V$.
	To check whether the subsystem is saturated, we require it to have more probability than a certain threshold, which depends on $k$ and $|V|$. Details can be found in \referapp{}.
\end{proof}

\begin{remark}
	NP-completeness of transition-minimal and size-minimal versions of the witness problem for acyclic DTMCs follows along the same lines, where only the sizes and thresholds for the subsystems need to be adapted.
\end{remark}

However, DTMCs whose underlying graph is a tree permit an efficient algorithm for computing minimal witnesses (for the proof see \citeapp{thm: tree-shaped algorithm}).

 \begin{restatable}{proposition}{treeshapeddtmc}
 	Minimal witnesses in tree-shaped DTMCs can be computed in polynomial time.
 \end{restatable}
\begin{proof}[Sketch]
	The algorithm first transforms the DTMC at hand into a binary (tree-shaped) DTMC, and then works bottom up by storing for each state the highest probability that can be obtained with a subsystem of size $k$, for all $k$ up to the size of the subtree.
\end{proof}

\section{Relating Farkas certificates and minimal witnesses}
\label{sec:relation}

In this section we establish a strong connection between Farkas certificates on the one hand and witnesses for probabilistic reachability constraints on the other hand. We first note that the set of Farkas certificates for non-strict lower bounds forms a polytope, i.e., a bounded polyhedron.

\begin{restatable}[Polytopes of Farkas certificates]{lemma}{polytope}\label{lem:polytope}
	Let $\mathcal{M} = (S_\all,\Act,s_0,\Pb)$ be an MDP as in Setting \ref{setting:mdp} and consider $\Ab\in\R^{\M\times S}$ and $\bb\in\R^S$ introduced there. Then for every $\lambda \in [0,1]$ the polyhedra
	\begin{gather*}
		\P^{\min}(\lambda) = \{ \zb\in \R^{S}\mid \Ab\zb\leq \bb \land \zb(s_0) \geq \lambda \}\\
		\P^{\max}(\lambda) = \{ \yb\in \R^{\M}\mid \yb\geq 0\land \yb\Ab\leq \delta_{s_0} \land \yb\bb \geq\lambda \}
	\end{gather*}
	are both polytopes, called \emph{the polytopes of Farkas certificates}.
\end{restatable}

\begin{remark}
	For any vector $\vb\in \R^n$ the support is defined as $\supp(\vb) =\{i\in\{1,...,n\}\mid \vb_i >0\}$, and analogously for the vector spaces $\R^S$ and $\R^{\M}$.
  As our connection between subsystems of $\M$ and points in $\P^{\min}(\lambda)$ is based on taking the support, we restrict our attention to the subpolytope $\P^{\min}_{\geq 0}(\lambda)  = \P^{\min}(\lambda)\cap \R^S_{\geq 0}$.
\end{remark}

\begin{notation}
	Given an MDP $\mathcal{M} = (S_\all,\Act,s_0,\Pb)$ as in Setting \ref{setting:mdp} and a subset $R\subseteq \M$, where $\M$ also denotes the state-action pairs (compare with \Cref{sec:prelim}). We let $\M_R = (S'_\all,\Act,s_0,\Pb')$ be the subsystem where, roughly speaking, the state-action pairs in $R$ \emph{remain}. More precisely, let
	\begin{align*}
		S'_\all &= \{ s\in S\mid \exists \alpha\in\Act. \; (s, \alpha)\in R\}\cup\{\goal, \fail\}\\
		\Pb'(s,\alpha, t) &= \begin{cases} 
			\Pb(s,\alpha, t) & \mbox{if } (s, \alpha)\in R\text{ and } t\in S'_\all\setminus \{\fail\}\\
			1-\sum_{t\in S'_\all\setminus\{\fail\}}\Pb(s, \alpha, t) & \mbox{if } (s, \alpha)\in R\text{ and } t=\fail\\
			1 & \mbox{if } (s, \alpha)\notin R, \alpha\in\Act(s) \text{ and } t=\fail\\
			0 & \mbox{else } 
		\end{cases}
		\end{align*}
	For $R\subseteq S$ we set $\M_R = \M_{R'}$ for $R' = \bigcup_{s\in R} \{s\}\times \Act(s)$.
\end{notation}

\begin{restatable}[Farkas certificates yield witnesses]{theorem}{mainthm}\label{thm:MCS-polytope}
 	Let $\mathcal{M}$ be an MDP as in Setting \ref{setting:mdp} and $\lambda\in[0,1]$. Then for a set $R\subseteq S$ the following statements are equivalent:
	\begin{enumerate}[leftmargin = 13mm]
		\item The subsystem $\M_R$ is a witness for $\prb^{\min}_{\M, s_0}(\lozenge\goal) \geq\lambda$.
		\item There is a point $\pb$ in $\P^{\min}_{\geq 0}(\lambda)$ such that $\supp(\pb) \subseteq  R$.
		\item There is a vertex $\vb$ of $\P^{\min}_{\geq 0} (\lambda)$ such that $\supp(\vb)\subseteq R$.
	\end{enumerate}
	Moreover, for a set $R\subseteq \M$ the following statements are equivalent:
	\begin{enumerate}[leftmargin = 13mm, label=(\alph*)]
	\item The subsystem $\M_R$ is a witness for $\prb^{\max}_{\M, s_0}(\lozenge\goal) \geq\lambda$.
	\item There is a point $\pb$ in $\P^{\max}(\lambda)$ such that $\supp(\pb) \subseteq  R$.
	\item There is a vertex $\vb$ of $\P^{\max}(\lambda)$ such that $\supp(\vb)\subseteq R$.
	\end{enumerate}
\end{restatable}
\noindent One consequence of~\Cref{thm:MCS-polytope} is that every MD-scheduler $\S$ with $\Pr^{\S}_{s_0}(\lozenge\goal) \geq\lambda$ corresponds to a point in $\P^{\max}(\lambda)$, i.e. to a certificate for $\prb^{\max}_{\M, s_0}(\lozenge\goal) \geq\lambda$.
\begin{restatable}[Detecting minimal witnesses by vertices of $\mathcal P$]{corollary}{verticestoMW}\label{cor:main corollary}
	Let $\mathcal{M} = (S_\all,\Act,s_0,\Pb)$ be an MDP as in Setting \ref{setting:mdp} and $\lambda\in[0,1]$. Then a vertex $\vb$ of $\P^{\min}_{\geq 0}(\lambda)$ has a maximal number of zeros among all vertices of $\P^{\min}_{\geq 0}(\lambda)$ if and only if $\M_{\supp(\vb)} $ is a minimal witness for $\prb^{\min}_{s_0}(\lozenge\goal)\geq \lambda$.
	
	Dually, a vertex $\vb$ of $\P^{\max}(\lambda)$ has a maximal number of zeros among all vertices of $\P^{\max}(\lambda)$ if and only if all of the following hold:
	\begin{enumerate}[leftmargin =10mm]
		\item $\M_{\supp(\vb)} = (S'_\all,\Act,s_0,\Pb')$ is a minimal witness for $\prb^{\max}_{s_0}(\lozenge\goal)\geq \lambda$,
		\item for every $s\in S'$ there is precisely one $\alpha\in\Act(s)$ with $(s, \alpha)\in \supp(\vb)$,
		\item the corresponding map $\S\colon S'\to\Act$ is an MD-scheduler on $\M_{\supp(\vb)}$ with $\Pr^\S_{s_0}(\lozenge\goal) \geq \lambda$.
	\end{enumerate}
\end{restatable}

\section{Computing witnessing subsystems}
\label{sec:computing}

In this section we use the results of \Cref{sec:relation} to derive two algorithms for the computation of minimal witnesses for reachability constraints in MDPs.
As the problem is NP-hard, we also present a heuristic approach aimed at computing small witnessing subsystems.

\subsubsection{Vertex enumeration.}
\Cref{cor:main corollary} gives rise to the following approach of computing minimal witnessing subsystems: enumerate all vertices in the corresponding polytope and choose one with a maximal amount of zeros.
Vertex enumeration of polytopes has been studied extensively~\cite{AvisF93, AvisF92, Balinski61, BremnerFM98, BussiekL98, Dyer83, Dyer77, Mattheiss73, Provan1994, FukudaLM97, FukudaP95} and has been shown to be computationally hard \cite[Corollary 2]{Khachiyan2008}.

First experiments that we have conducted with the \texttt{SageMath}\footnote{\url{http://www.sagemath.org/}} toolkit which supports vertex enumeration have not scaled well in the dimension, which in our case is the number of states in the original system.
Also, we found no tool support for vertex enumeration that is able to handle sparse matrices, which is essential for bigger benchmarks.

\subsubsection{Mixed integer linear programming.}
An approach that computes minimal witnesses to the threshold problem $\prb^{\max}_{s_0}(\lozenge\goal)\geq \lambda$ using mixed integer linear programs (MILP) was presented in~\cite{WimmerJABK12,WimmerJAKB14}. 
Using the following lemma, we can derive MILP formulations from our polytope formulations.
\begin{restatable}{lemma}{MILPformulation}\label{lem:MILP}
  \label{lem:MILPformulation}
  Let $\mathcal{P} = \{ \xb \mid \Ab \xb \leq \bb, \xb \geq 0 \} \subseteq \mathbb{R}^n$ be a polytope and $K \geq 0$ be such that for all $\pb \in \mathcal{P}$ and $1\leq i\leq n$ we have  $\pb(i) \leq K$.
  Consider the MILP
  \[ \min \sum_{1 \leq i \leq n} \ind(i) \quad\text{s.t.}\quad \xb \in \mathcal{P}, \quad \xb \leq K \cdot \ind, \quad\ind(i)\in\{0,1\}\]

  Then a vector $(\ind,\xb)$ is an optimal solution of this MILP if and only if $\xb$ is a point in $\mathcal{P}$ with a maximal number of zeros.
\end{restatable}

For $\mathcal{P}^{\min}_{\geq 0}(\lambda)$ we can use \Cref{lem:monotonicity} to derive that $K = 1$ is a viable bound.
By invoking again \Cref{cor:main corollary}, this means that a solution $(\zb, \ind)$ of the MILP
 \[ \min \sum_{s \in S} \ind(s) \;\text{ s.t. }\; \zb\in  \mathcal{P}^{\min}_{\geq 0}(\lambda), \quad \zb \leq \ind, \quad\ind(i)\in\{0,1\}\]
encodes a minimal witnessing subsystem in the integral variables $\ind$. This MILP was used in~\cite{WimmerJABK12,WimmerJAKB14} for the computation of minimal witnessing subsystems of DTMCs .
 
An upper bound $K$ as in \Cref{lem:MILP} for $\mathcal{P}^{\max}(\lambda)$ can be found in polynomial time by taking the objective value of an optimal solution to the LP
\[\max \sum_{(s,\alpha) \in \M} \yb(s,\alpha) \;\text{ s.t. }\; \yb\in \mathcal{P}^{\max}(\lambda) \]

\begin{remark}
\label{rem:ltlsubsys}
To compute minimal witnesses for $\prb^{\max}_{s_0}(\lozenge \goal) \geq \lambda$,~\cite{WimmerJABK12,WimmerJAKB14} (witnesses for $\prb^{\min}_{s_0}(\lozenge \goal) \geq \lambda$ were not considered) propose the MILP with objective: $\min \; \sum_{(s,\alpha) \in \M} \ind(s,\alpha)$, subject to the conditions
\begin{align}
 \forall (s,\alpha) \in \M. & \;\;\; \zb(s) \leq 1 - \ind(s,\alpha) + \sum_{s' \in S} \Pb(s,\alpha,s') \cdot \zb(s') + \bb(s)\\
  \forall s \in S.&\;\;\;\zb(s) \leq \sum_{\alpha \in \Act(s)} \ind(s,\alpha), \;\;\; \zb(s_{0}) \geq \lambda 
\end{align}
where $\ind(s,\alpha)$ are binary integer variables.
It was implemented in the tool \texttt{ltlsubsys}.
The idea is to directly encode a scheduler in the set of equations $\Ab \zb \leq \bb$ using $\ind$.
In~\cite{WimmerJABK12,WimmerJAKB14} a number of additional redundant constraints are given to guide the search.
In contrast to~\cite{WimmerJABK12,WimmerJAKB14} we do not need to handle so-called \emph{problematic states}, as our precondition $\prb^{\min}_s(\lozenge(\goal \lor \fail)) > 0$ guarantees that no such states exist.

\iffalse
Inequalities in (6.1) that do not belong to the subsystem representing the induced scheduler, i.e. the ones where $\ind(s,\alpha) = 0$, are ``discharged'' as $\zb(s) \leq 1 + N$ holds for any positive $N$.
This again uses the insight that $1$ is an upper bound for $\zb(s)$.
In the other inequalities, where $\ind(s,\alpha) = 1$, we get exactly the corresponding inequalities of $\Ab \zb \leq \bb$.
The constraints in (6.2) guarantee that at least one $\ind(s,\alpha)$ is $1$ if $\zb(s) > 0$, and (6.3) ensures that the probability in $s_0$ under the given scheduler is above the threshold $\lambda$.
We have omitted the handling of what were called \emph{problematic states} in~\cite{WimmerJABK12,WimmerJAKB14}, as our assumption $\prb^{\min}_s(\lozenge(\goal \lor \fail)) > 0$ for all states $s$ guarantees that such states do not exist. In \Cref{sec:case studies} we give an experimental evaluation of these two different MILP formulations for $\prb^{\max}_{s_0}(\lozenge \goal) \geq \lambda$.
\fi
\end{remark}

\subsubsection{$k$-step quotient sum ($\qs_k$) heuristics.}

\newcommand{\qs}{\operatorname{QS}}

Approximating the maximal number of zeros in a polytope is computationally hard in general~\cite{AmaldiK98}.
We now derive a heuristic approach for this problem called \emph{quotient sum heuristic} which is based on iteratively solving LPs over the polytope, where the objective function for each iteration depends on an optimal solution of the previous LP. More precisely, we take $\mathbf{o}_1 = (1,\ldots,1)$ and take an optimal solution $\qs_1$ of the LP $\min \mathbf{o}_1\cdot\yb \;\text{ s.t. } \; \yb\in\mathcal{P}^{\max}(\lambda)$.
Many entries in $\qs_1$ may be small, but still greater than zero.
In order to push as many of the small values of $\qs_1$ to zero, we define a new objective function by
\begin{equation}\label{eq:QS}
\mathbf{o}_2(i) =
\begin{cases}
  1/\qs_1(i), & \text{if } \qs_1(i) > 0 \\
  C, & \text{if } \qs_1(i) = 0
\end{cases}
\end{equation}
where $C$ is a value that is greater than any value $1/\qs_1(i)$.
We now take a solution $\qs_2$ of the new LP $\min \mathbf{o}_2\cdot\yb \text{ s.t. } \yb\in\mathcal{P}^{\max}(\lambda)$
and form the next objective function $\mathbf{o}_3$ as in (\ref{eq:QS}).
Inductively this generates a sequence of objective functions $(\mathbf{o}_k)_{k\geq 1}$ and corresponding optimal solutions $(\qs_k)_{k
  \geq 1}$ in $\mathcal{P}^{\max / \min}(\lambda)$.
By~\Cref{thm:MCS-polytope} we can construct a witnessing subsystem with as many states as the number of non-zero entries in $\qs_k$.

\section{Experiments}\label{sec:case studies}
In this section we evaluate our MILP formulations and heuristics on a number of DTMC and MDP benchmarks from the \prism{} benchmark-suite~\cite{KwiatkowskaNP11,KwiatkowskaNP12}.
We compare our results with the tool \comics{}~\cite{JansenAVWKB12}, which implements heuristic approaches to compute small subsystems for DTMCs.
It has two modes: the \emph{local search} extends a given subsystem by short paths that carry much probability, whereas the \emph{global search} searches for the next most probable path from the initial state to $\goal$, and adds it to the subsystem.
Both approaches iteratively extend a subsystem until it carries more probability than the given threshold and thus have to compute the probability of the subsystem at each iteration.

All computations were performed on a computer with two Intel E5-2680 8 cores at \(2.70\)\,GHz running Linux, with a time bound of $30$ minutes, a memory bound of 100\,GB and with each benchmark instance having access to 4 cores.
For the LP and MILP instances we use the Gurobi solver, version 8.1.1~\cite{gurobi}.
The recorded times of our computations include the construction of the LPs/MILPs and are wall clock times.
Pre-processing steps, such as collapsing states that cannot reach $\goal$, are not counted in the time consumption.
For \comics{}, we use the time that is reported as counterexample generation time by the tool.

To validate our implementation, we used \prism{} to verify that the subsystems that we compute indeed satisfy the probability thresholds.
We noticed that for a few instances ($< 0.5\%$) \prism{} reported a deviation of less than $10^{-8}$, which can be explained by the fact that both \prism{} and the solvers that we use rely on floating-point arithmetic, which is approximate by nature.

Our implementation, together with the models we use and benchmark results can be found at \url{https://github.com/simonjantsch/farkas}.

\begin{figure}[h]
	\captionsetup{width=.75\linewidth, labelfont={bf, up}}
  \begin{subfigure}[t]{0.5\columnwidth}
    \includegraphics[width=6cm]{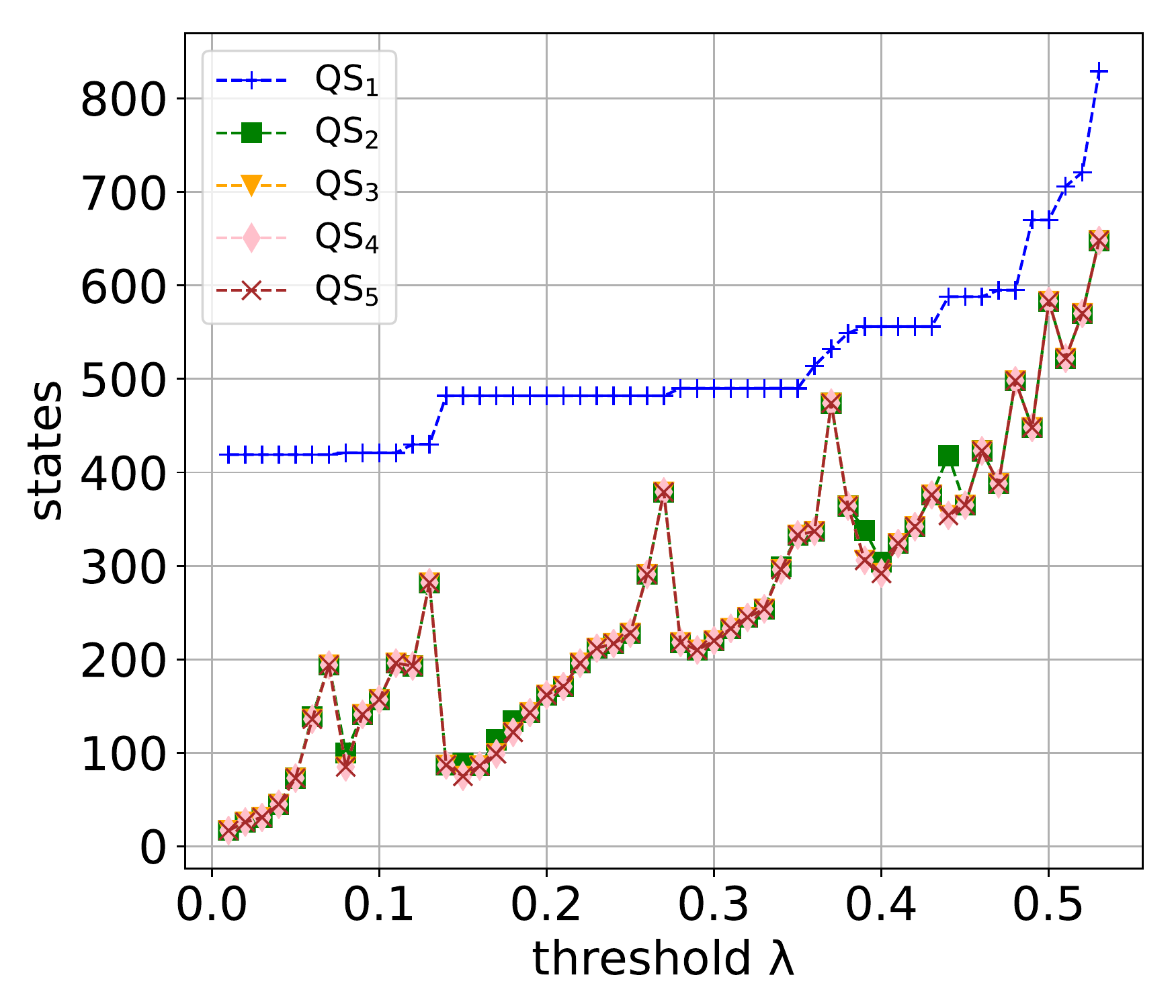}
    \subcaption{QS-heuristic applied to $\mathcal{P}^{\max}(\lambda)$.}
    \label{subfig:itpmax}
  \end{subfigure}
  \begin{subfigure}[t]{0.5\columnwidth}
    \includegraphics[width=6cm]{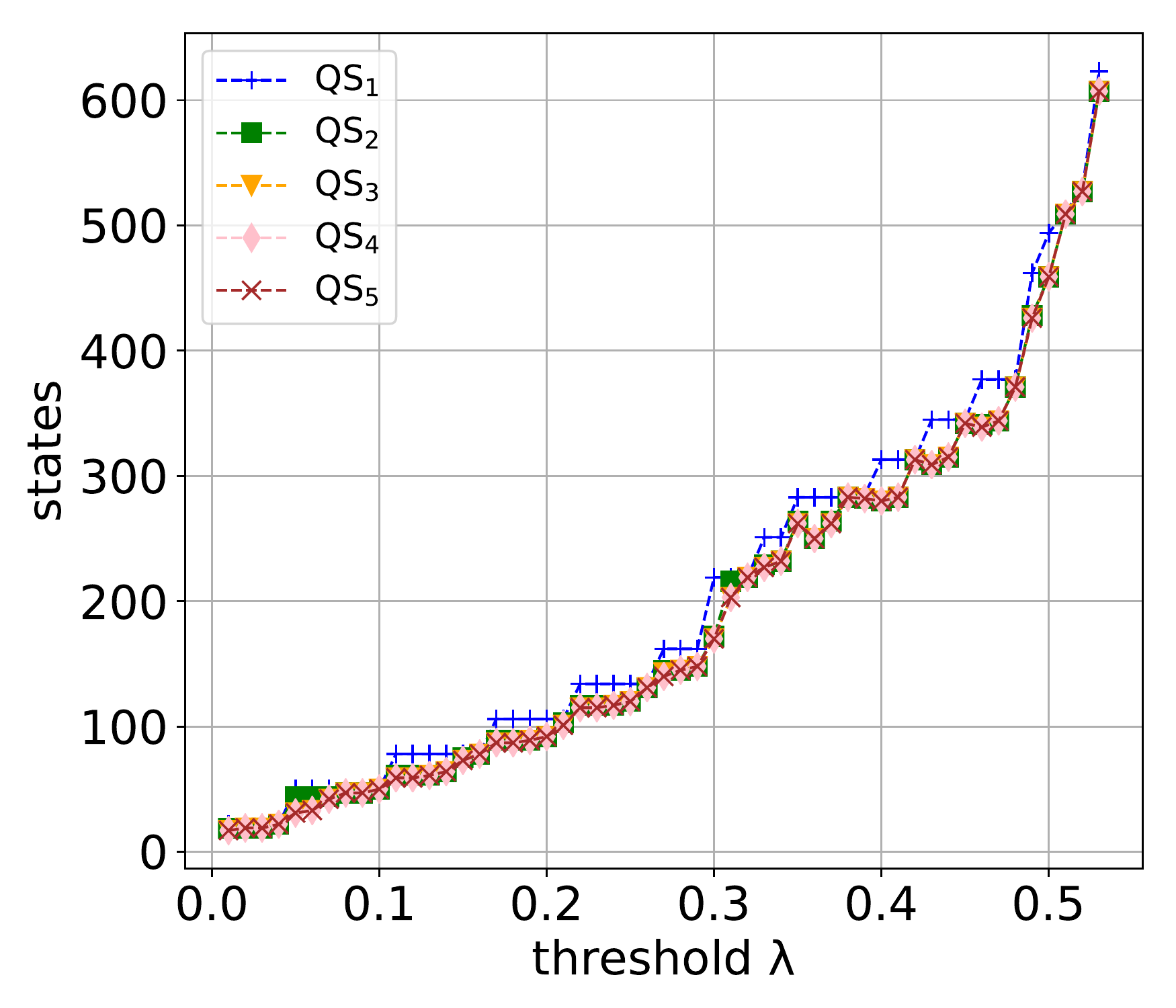}
    \subcaption{QS-heuristic applied to $\mathcal{P}^{\min}_{\geq 0}(\lambda)$.}
  \end{subfigure}
  \caption{crowds-2-8: comparing $\qs_k$ for growing $k$.}
  \label{fig:expheuriter}
\end{figure}

\begin{figure}[t]
	\captionsetup{width=.9\linewidth, labelfont={bf, up}}
	\begin{subfigure}[h]{1\columnwidth}
		\begin{subfigure}{0.5\columnwidth}
			\includegraphics[width=6cm]{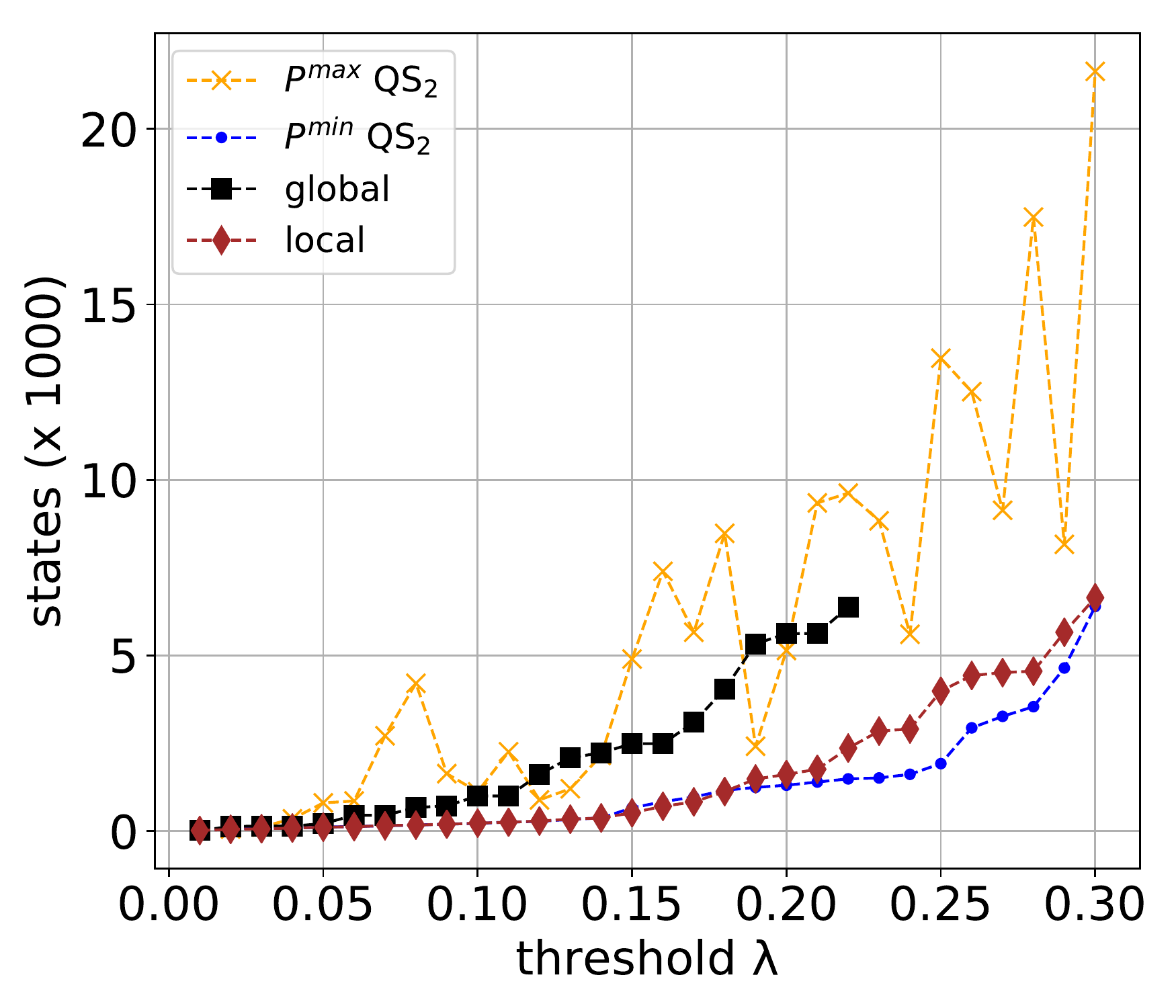}
		\end{subfigure}
		\begin{subfigure}{0.5\columnwidth}
			\includegraphics[width=6cm]{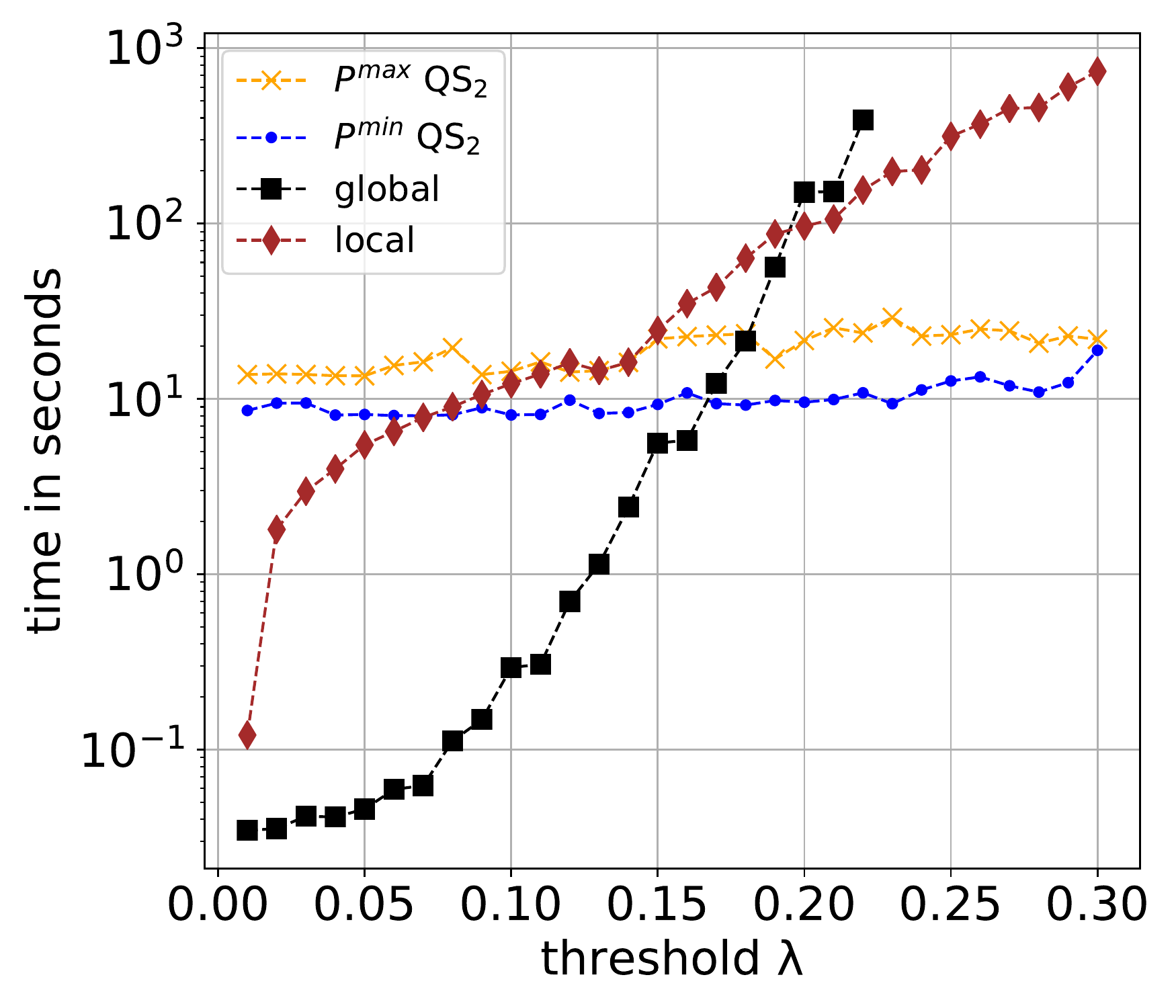}
		\end{subfigure}
		\captionsetup{width=.9\linewidth, labelfont={bf, up}}
		\subcaption{crowds-5-8 (46,873 states). \comics{}-global runs out of memory for $\lambda \geq 0.23$.}
		\label{fig:expheurcrowds}
	\end{subfigure}
	\vspace{0.1cm}
	\begin{subfigure}[h]{1\columnwidth}
		\begin{subfigure}{0.5\columnwidth}
			\includegraphics[width=6cm]{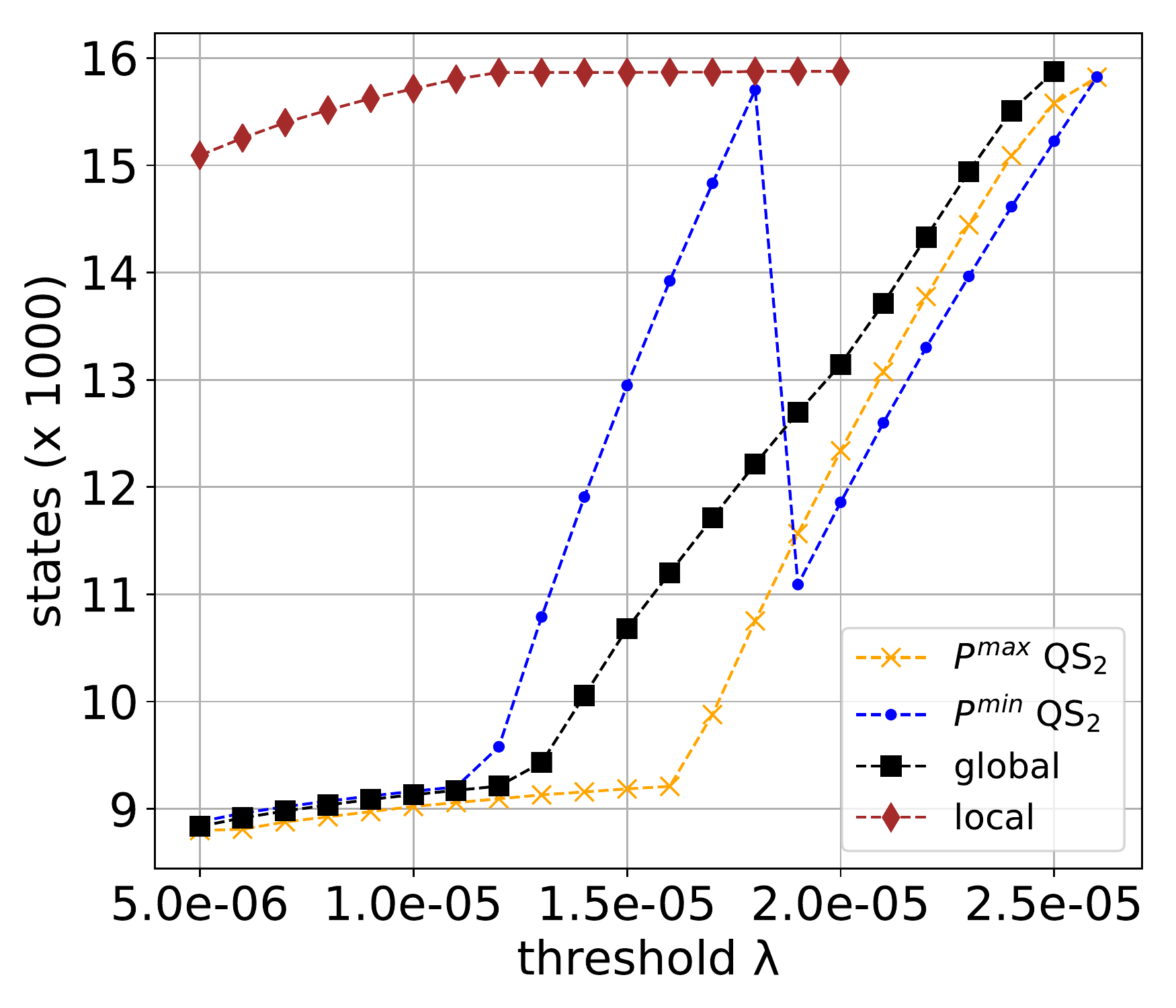}
		\end{subfigure}
		\begin{subfigure}{0.5\columnwidth}
			\includegraphics[width=6cm]{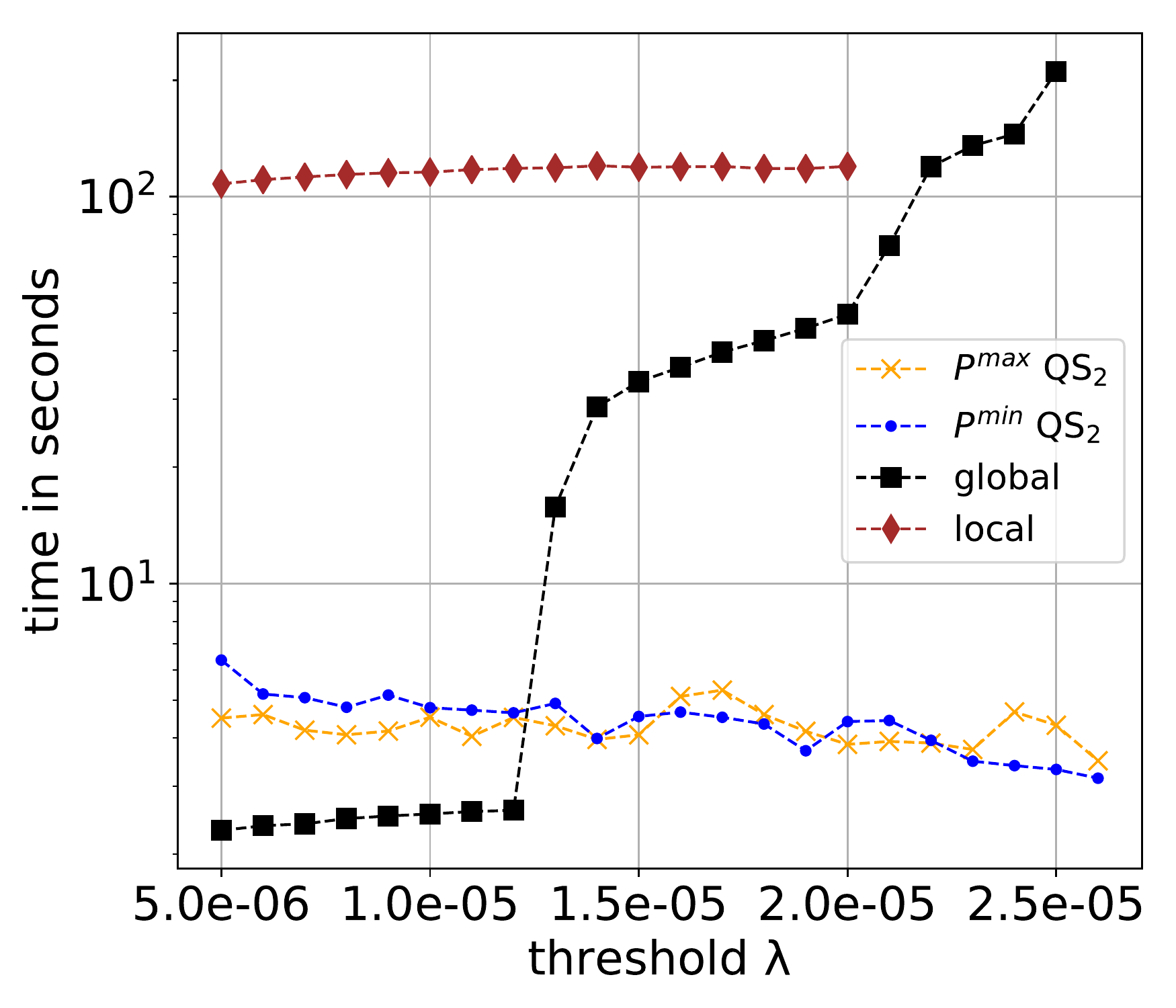}
		\end{subfigure}
		\captionsetup{width=.9\linewidth, labelfont={bf, up}}
		\subcaption{brp-512-2 (15,875 states). \comics{}-local reports an error for $\lambda \geq 2.1\cdot10^{-5}$ and \comics{}-global runs out of memory for $\lambda = 2.6 \cdot 10^{-5}$. }
		\label{fig:expheurbrp}
	\end{subfigure}
	\caption{Comparison of heuristic methods on DTMC benchmarks.}
	\label{fig:expheur}
\end{figure}

\subsubsection{DTMC benchmarks.}

As $\prb^{\max}$ and $\prb^{\min}$ coincide on DTMCs, we can use the heuristics and exact computations derived from either the $\mathcal{P}^{\max}$ or the $\mathcal{P}^{\min}_{\geq 0}$ polytope for DTMCs (in \comics{} we use the standard query $\Pr_{s_0}(\lozenge \goal) \geq \lambda$).
We consider two DTMC benchmarks: a model of the crowds-$N$-$K$ protocol~\cite{Shmatikov04,ReiterR98} for ensuring anonymous web browsing (with $N$ members and $K$ protocol runs) and a model of the bounded retransmission protocol~\cite{DArgenioJJL01,HelminkSV93} for file transfers (where brp-$N$-$K$ is the instance with $N$ chunks and $K$ retransmissions).

\Cref{fig:expheuriter} shows the effect of increasing the number of iterations of the $\qs$-heuristic for the model crowds-2-8.
While the first iteration (taking $\qs_2$ instead of $\qs_1$) has an impact on the number of states, more iterations do not improve the result significantly.
For $\qs_1$, the sizes of subsystems increase monotonically with growing $\lambda$. Starting with $\qs_2$ the results may, interestingly, have ``spikes'': increasing $\lambda$ can lead to smaller subsystems.

\Cref{fig:expheur} shows the results of the $\qs_2$-heuristic compared to the two modes of \comics{} for $\lambda$ that ranges between $0$ and the actual reachability probability of the model.
A general observation is that the runtime of the $\qs$-heuristic is independent of $\lambda$, whereas both modes of \comics{} use significantly more time with increasing $\lambda$.
The same observation can be done for memory consumption, which stayed below 200 MB for our heuristics.
Also, especially for crowds-5-8, one can see that relatively small subsystems are possible even for large $\lambda$.
The exact computations via MILPs hit the timeout for almost all instances.

In~\Cref{fig:expheur} it can be seen that the $\qs$ heuristics derived from the two polytopes $\P^{\max}$ and $\P_{\geq 0}^{\min}$ may produce different results.
However, for both models one of them gives monotonically growing subsystems and outperforms \comics{}.
While $\qs_2$ applied to $\mathcal{P}^{\min}_{\geq 0}$ performs better on crowds-5-8 (\Cref{fig:expheurcrowds}), it is the other way around on brp-512-2 (\Cref{fig:expheurbrp}).
In future work we intend to investigate what properties determine which of the two formulations performs better for a given DTMC.

\subsubsection{MDP benchmarks.}
\begin{figure}[tbp]
  \captionsetup{width=.75\linewidth, labelfont={bf, up}}
  \begin{subfigure}[t]{0.5\columnwidth}
    \includegraphics[width=6cm]{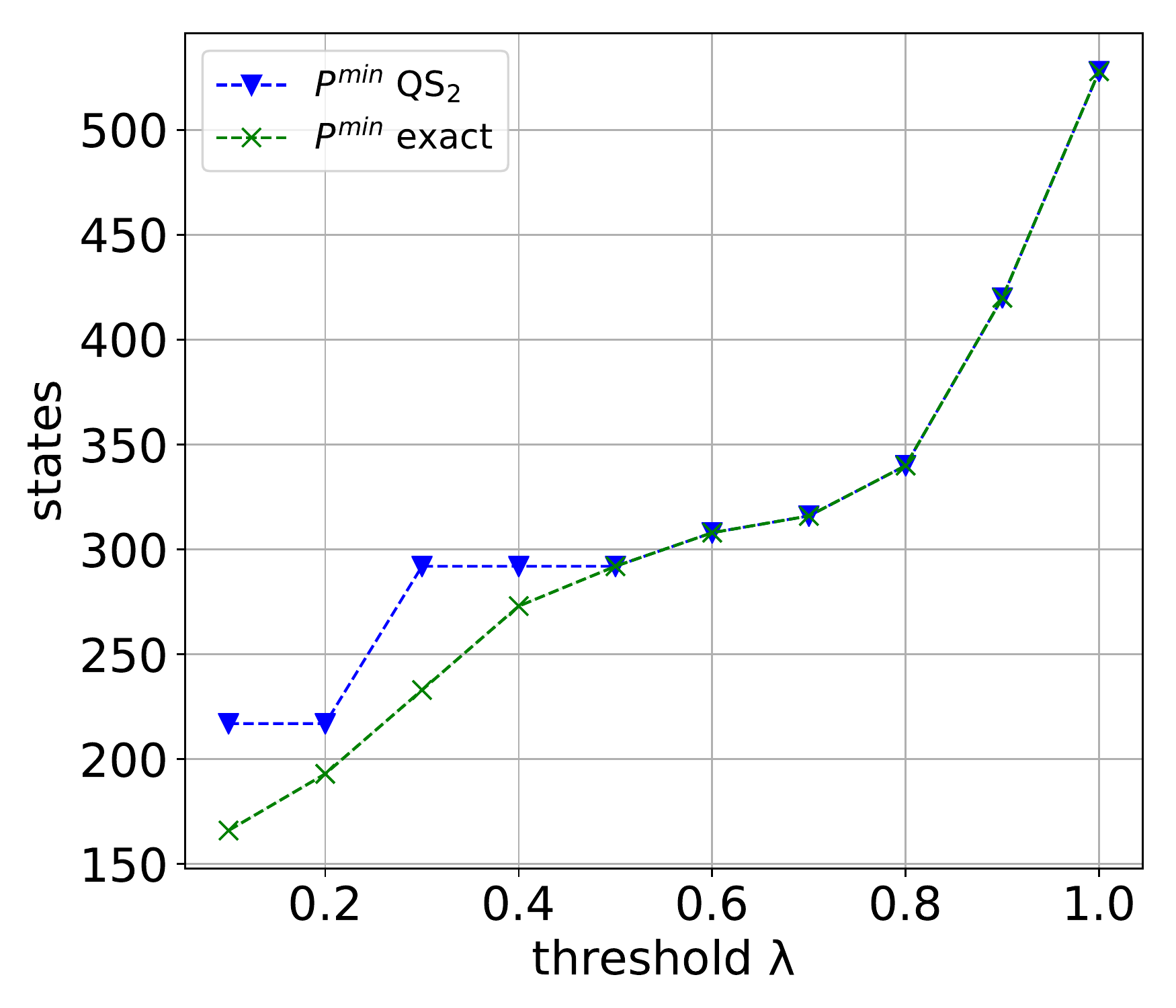}
    \subcaption{Witnesses for $\prb^{\min}_{s_0}(\lozenge \goal) \geq \lambda$.}
    \label{subfig:mdpconsprmin}
  \end{subfigure}
  \begin{subfigure}[t]{0.5\columnwidth}
    \includegraphics[width=6cm]{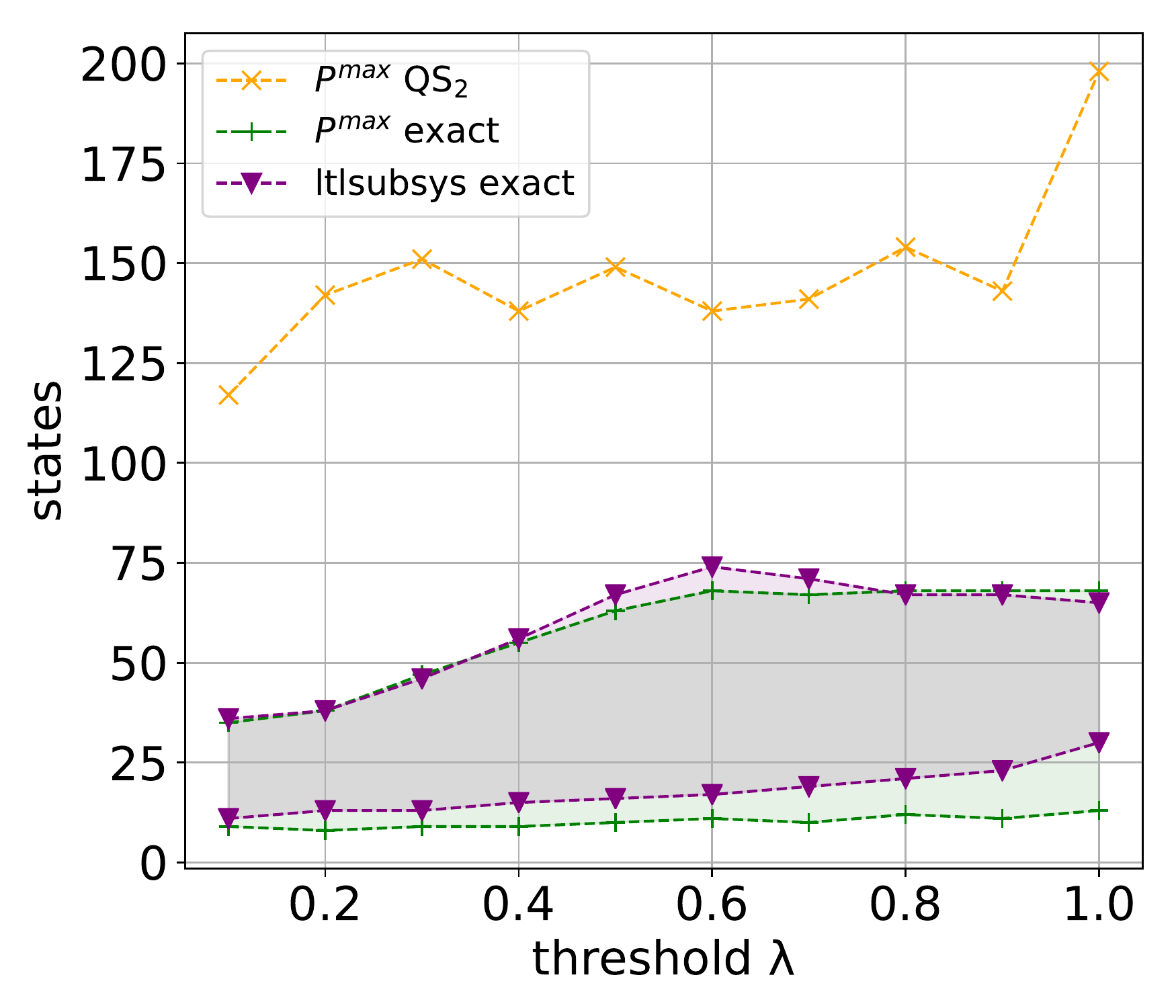}
    \subcaption{Witnesses for $\prb^{\max}_{s_0}(\lozenge \goal) \geq \lambda$.}
    \label{subfig:mdpconsprmax}
  \end{subfigure}
  \caption{MDP benchmark: consensus-2-4 (528 states)}
  \label{fig:expmdpconsensus}
\end{figure}

We consider two MDP models: the randomized consensus-$N$-$K$ protocol of~\cite{KwiatkowskaNS01,AspnesH90} (with $N$ processes and a bound $K$ on the random walk) and the CSMA-$N$-$K$ protocol for data channels~\cite{KwiatkowskaNSW07} (where $N$ is the number of stations, and $K$ is the maximal backoff count).
The results of both heuristic and exact computations can be seen in \Cref{fig:expmdpconsensus} and \Cref{fig:expmdpcsma}.
Whereas the heuristics all needed less than $5$ minutes, all MILP instances ran into the timeout except for the ones in \Cref{subfig:mdpconsprmin}.
Whenever a MILP instance could not be solved optimally in $30$ minutes, we plot both the found upper and lower bound, with the region in between shaded.
It should be noted that the condition $\prb^{\min}(\lozenge (\goal \lor \fail))$ holds for the instances of these models, and reachability properties, that we consider.

\begin{figure}[tbp]
  \captionsetup{width=.75\linewidth, labelfont={bf, up}}
  \begin{subfigure}[t]{0.5\columnwidth}
    \includegraphics[width=6cm]{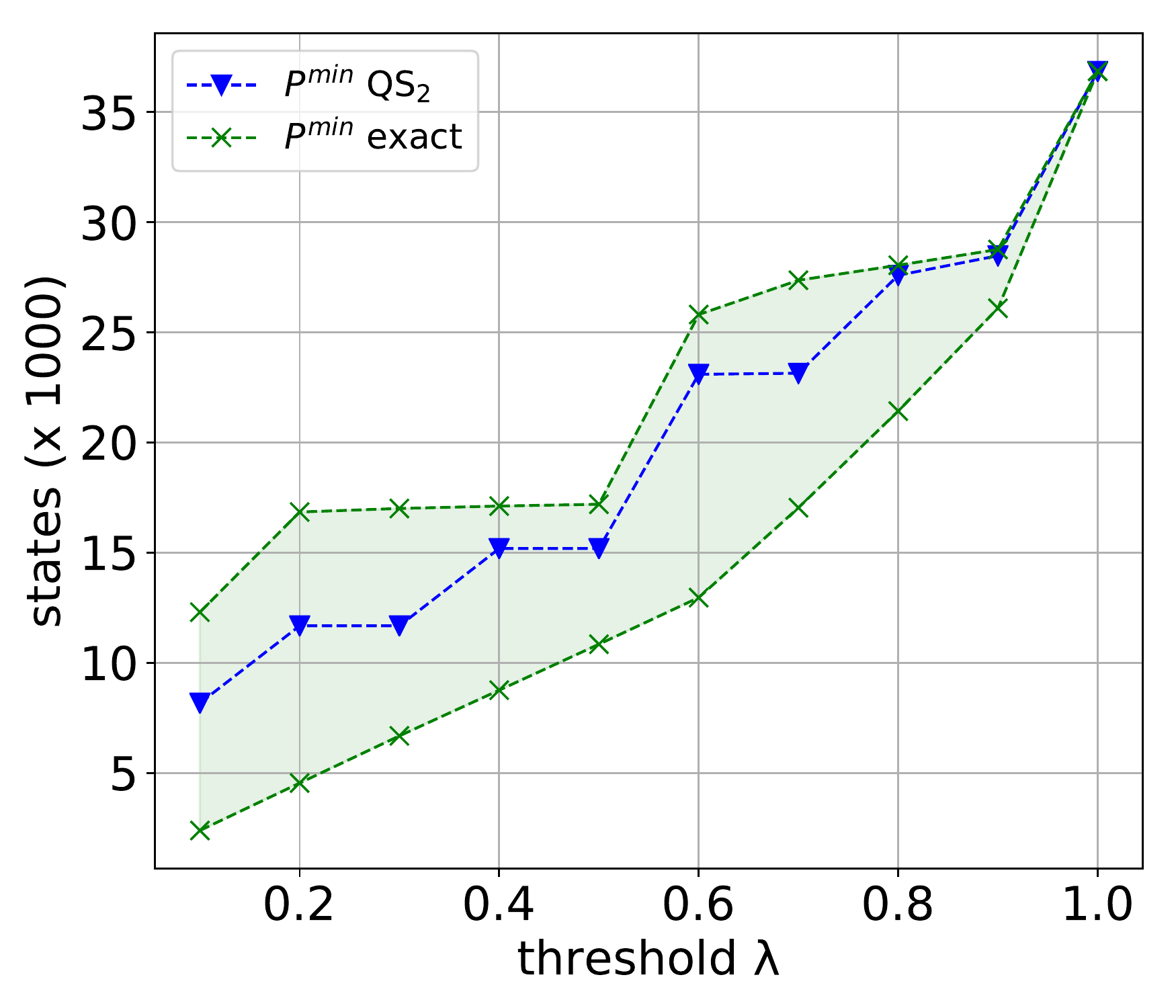}
    \subcaption{Witnesses for $\prb^{\min}_{s_0}(\lozenge \goal) \geq \lambda$.}
  \end{subfigure}
  \begin{subfigure}[t]{0.5\columnwidth}
    \includegraphics[width=6cm]{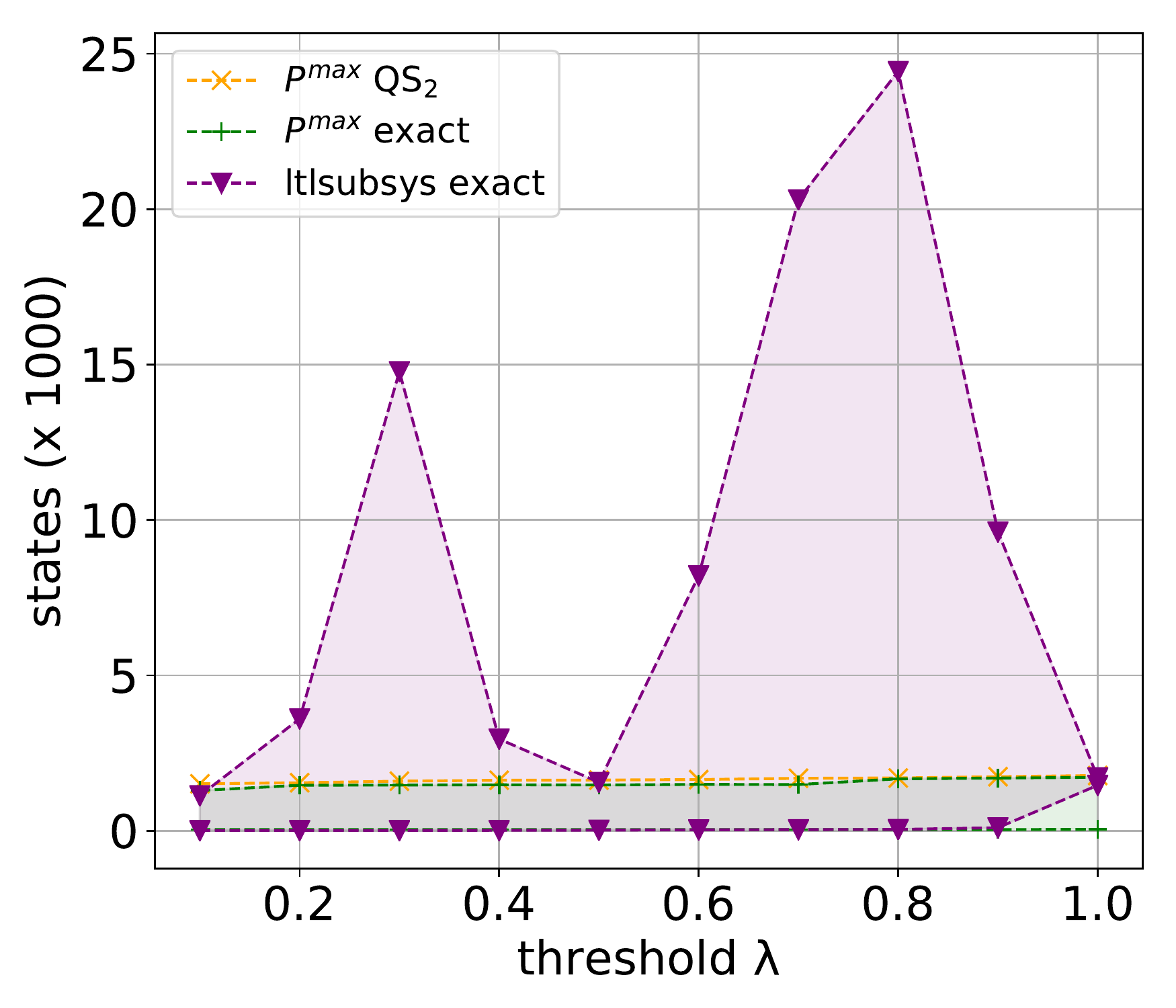}
    \subcaption{Witnesses for $\prb^{\max}_{s_0}(\lozenge \goal) \geq \lambda$.}
    \label{subfig:mdpcsmaprmax}
  \end{subfigure}
  \caption{MDP benchmark: CSMA-3-2 (36,850 states)}
  \label{fig:expmdpcsma}
\end{figure}

The comparison between the MILP formulation that we derived from $\mathcal{P}^{\max}(\lambda)$ and the one presented in~\cite{WimmerJABK12,WimmerJAKB14} (labeled by \texttt{ltlsubsys}, see also \Cref{sec:computing}) shows that both compute comparable upper and lower bounds in~\Cref{subfig:mdpconsprmax}, whereas \texttt{ltlsubsys} found worse upper bounds in~\Cref{subfig:mdpcsmaprmax}.
In all instances apart from \Cref{subfig:mdpconsprmax} the corresponding $\qs_2$ heuristics performs well and generates subsystems that are as good, or better, than the best upper bounds computed by the MILPs in 30 minutes.
As expected, the witnessing subsystems for $\prb^{\min}_{s_0}(\lozenge \goal) \geq \lambda$ tend to the entire state space as $\lambda$ tends to the actual value $\prb^{\min}_{s_0}(\lozenge \goal)$ (which is $1$ in these two models).
However, subsystems for $\prb^{\max}_{s_0}(\lozenge \goal) \geq \lambda$ may be substantially smaller even for large $\lambda$.

\section{Conclusion}
In this paper we brought together two a priori unrelated notions in the context of probabilistic reachability constraints: on the one hand Farkas certificates, which are vectors satisfying certain linear inequalities that we derive using MDP-specific variants of Farkas' Lemma, and on the other hand witnessing subsystems, which provide insight into which parts of the system are essential for the satisfaction of the considered property. This connection reduces the computation of minimal (respectively, small) witnessing subsystems to finding a Farkas certificate with a maximal (respectively, large) number of zeros.
Furthermore, it leads to a unified notion of witnessing subsystem for $\prb^{\max}_{s_0}(\lozenge \goal) \geq \lambda$ and $\prb^{\min}_{s_0}(\lozenge \goal) \geq \lambda$.

We showed that the decision version of computing minimal witnessing subsystems is NP-complete for acyclic DTMCs and introduced heuristics for the computation of small witnesses based on Farkas certificates.
Experiments of the heuristics exhibited competitive results compared to the approach implemented in \comics{} and showed that they scale well with the system size and threshold.
As expected, computing minimal subsystems using the derived MILP formulations consumed significantly more time than the heuristics and often triggered timeouts.
The upper and lower bounds that were computed in the given time by the new MILP formulation for $\prb^{\max}_{s_0}(\lozenge \goal) \geq \lambda$ were comparable to known techniques.

We have considered MDPs in which the probability to reach $\goal$ or $\fail$ is positive under each scheduler.
In future work, we plan to extend our techniques to weaken this assumption.
Exploring how vertex enumeration techniques could be adapted to the MDP-specific form of the Farkas polytopes is another interesting line of future work.
We also plan to implement a tool for working with Farkas certificates in practice, which encompasses their generation as well as their independent validation.

\bibliographystyle{splncs04}
\bibliography{lit}

\iflongversion

\newpage
\appendix

\section{Proofs of \Cref{sec:prelim}}

\begin{lemma}[Almost sure reachability]\label{lem: almost sure BSCC}
	Let $\M$ be an MDP as in Setting \ref{setting:mdp}. Then we have $\prb^{\min}_s(\lozenge (\goal\lor\fail)) = 1$ for every $s\in S$.
\end{lemma}
\begin{proof}
	Take any MD-scheduler $\S$ on $\M$ and fix some $s\in S$. We may assume that every state is reachable from $s$, otherwise we restrict $\M^{\S}$ to the set of states reachable from $s$. Since $\prb^{\min}_t(\lozenge (\goal\lor\fail)) > 0$, we also have $\Pr^{\S}_t(\lozenge (\goal\lor\fail)) >0$ for every $t\in T$. This implies that the only state subsets which can be a BSCC in $\M^{\S}$ are $\{\fail\}$ and $\{\goal\}$. But almost every path in a DTMC reaches a BSCC, see e.g. \cite[Theorem 10.27]{BaierK2008}, and thus in fact $\Pr^{\S}_s(\lozenge (\goal\lor\fail)) = 1$. Since we took an arbitrary MD-scheduler and the minimal probability is attained by an MD-scheduler, we have $\prb^{\min}_s(\lozenge (\goal\lor\fail))=1$.
\qed
\end{proof}

\prminmaxlp*
\begin{proof}
	We heavily build on \cite[Section 10.6.1]{BaierK2008}, but give the argument explicitly for $\prb^{\min}(\lozenge\goal)$ due to some relaxations we made to the linear programs. It is easily verified that $\Ab\cdot\prb^{\min}(\lozenge\goal)\leq \bb$. 
	
	We first prove that the LP 
  \begin{equation}\label{eq:LPforMax}
  	\max \, \db \cdot \zb \;\text{ s.t. } \; \Ab \zb \leq \bb \quad
  \end{equation}
  is bounded. If this was not the case, then there would exist $\zb_0, \zb_1\in \R^S$ with $\zb_1\neq 0$ such that $\Ab(\zb_0+t\zb_1)\leq \bb$ for all $t\in [0,\infty)$ and $\zb_1\db > 0$. From this we see that $\Ab\zb_1\leq 0$. Now let $m = \max\{\zb_1(s)\mid s\in S\}$ and let $S_{\max} = \{s\in S\mid \zb_1(s) = m\}$. Since $\zb_1\db > 0$, we must have $m>0$. Now fix some arbitrary $s_{\max}\in S_{\max}$. Then because of $\Ab\zb_1\leq 0$ we have for all $\alpha\in \Act(s)$	
	\[m= \zb_1(s_{\max}) \leq \sum_{t\in S} \Pb(s_{\max},\alpha,t)\cdot\zb_1(t)\leq \sum_{t\in S} \Pb(s_{\max},\alpha,t)\cdot m \leq m \]
	and hence everywhere equality. Since $m> 0$, we have $\sum_{t\in S} \Pb(s_{\max},\alpha,t) = 1$, meaning that there is no transition from $s_{\max}$ to $\fail$ or $\goal$, and $\zb_1(t) = m$, meaning that the same applies to all successors of $s_{\max}$. By induction one sees that $\goal$ and $\fail$ are therefore not reachable from $s_{\max}$, a contradiction to the assumption $\prb^{\min}_{s_{\max}}(\lozenge (\goal\lor\fail)) > 0$.
	
	Next we show that any $\zb\in \R^S$ with $\Ab\zb\leq \bb$ must have a vanishing entry for an $s'\in S$ such that $\prb^{\min}_{\M, s'}(\lozenge\goal) = 0$. Take an MD-scheduler $\S\colon S\to \Act$ on $\M$ with $\Pr^{\S}_{\M, s}(\lozenge\goal) = \prb^{\min}_{\M, s}(\lozenge\goal)$ for all $s\in S$. Let $T\subseteq S$ denote the set of states that are reachable in $\M^{\S}$ from the given state $s'$ with $\prb^{\min}_{\M, s'}(\lozenge\goal) =\Pr^{\S}_{s'}(\lozenge\goal)= 0$. Then there is no transition from any $t\in T$ to $\goal$.
	
	 Let $\Qb\in \R^{T\times T}$ be the transition matrix of $\M^{\S}$ restricted to $T$. Let $\zb'$ be the projection of $\zb\in \R^S$ to $\R^T$. Then from $\Ab\zb\leq\bb$ and the fact that $\bb(t,\S(t))= 0$ for every $t\in T$, we get $\zb'\leq \Qb\zb'$ and by induction $\zb'\leq \Qb^n\zb'$. By \Cref{lem: almost sure BSCC} almost every path in $\M^{\S}$ reaches $\fail$ or $\goal$. But $\Qb^n(s,t)$ is the probability to reach $t$ from $s$ in exactly $n$ steps. Therefore we must have $\Qb^n\to0$ as $n\to\infty$, and in particular $\zb(s') = \zb'(s') = 0$.
	
	So far, we have argued that the LP (\ref{eq:LPforMax}) has a solution, say $\zb^*$, and that $\zb^*(s) = 0$ for all $s\in S$ with $\prb^{\min}_{\M, s}(\lozenge\goal)= 0$. It is easy to see that for all states we must have
	\[ \zb^*(s) = \min_{\alpha\in\Act(s)} \left\{\sum_{t\in S} \Pb(s,\alpha,t)\cdot \zb^*(t) + \bb(s,\alpha) \right\} \]
	since we could otherwise increase $\zb^*(s)$ without changing $\zb^*$ elsewhere and obtain a better solution. From \cite[Theorem 10.109]{BaierK2008}, it now follows that we must have $\zb^* = \prb^{\min}(\lozenge\goal)$.
	
	The proof is very similar for $\prb^{\max}(\lozenge\goal)$.
\qed\end{proof}

\section{Proofs of \Cref{sec:certificates}}
\monotonicityprminprmax*
\begin{proof}
	This is a simple consequence of \Cref{prop:LP version of pr_min}. Assume that there was $\zb$ with $\Ab\zb\leq \bb$ and one entry would satisfy $\zb(s) > \prb^{\min}_s(\lozenge\goal)$. Then choose $\db \in \R^S$ with 
	\[ \db(s) = \frac{\sum_{t\neq s} |\zb(t) -\prb^{\min}_t(\lozenge\goal)| + 1}{\zb(s) - \prb^{\min}_s(\lozenge\goal)} \] 
	and $\db(t) = 1$ for all $t\neq s$. Then we calculate 
	\begin{align*}
		\db\cdot\left(\zb -\prb^{\min}(\lozenge\goal)\right) & = \sum_{t\neq s} \zb(t) -\prb^{\min}_t(\lozenge\goal) + \db(s) \cdot \left(\zb(s) - \prb^{\min}_s(\lozenge\goal)\right)\\
					& =  \sum_{t\neq s} \zb(t) -\prb^{\min}_t(\lozenge\goal) +  \sum_{t\neq s} |\zb(t) -\prb^{\min}_t(\lozenge\goal)| + 1\\
					&> 0.
	\end{align*}
	Thus we get $\db\cdot\zb > \db\cdot \prb^{\min}(\lozenge\goal)$ in contradiction to \Cref{prop:LP version of pr_min}.
\qed\end{proof}

\begin{remark}[$\yb$-vectors, schedulers and frequencies]\label{rem:frequencies}
	In the equivalences used to derive \Cref{prop:certificates for ex quantified}, we could replace $\zb\in \R^S_{\geq 0}$ with $\zb\in \R^S$, leading to the equivalent formulation
	\begin{align*} 
	\neg \,\exists \zb\in\R^S, z^*\in\R .\, \,\;  \begin{pmatrix}
	- \Ab & \bb\\
	0 \ldots 0 & 1
	\end{pmatrix}
	\begin{pmatrix}
	\zb \\
	z^*
	\end{pmatrix} \geq 0 \land
	\begin{pmatrix}
	-\delta_{s_0} &   \lambda
	\end{pmatrix} 
	\begin{pmatrix}
	\zb \\
	z^*
	\end{pmatrix} < 0  
	\end{align*}
	which is by a variant of Farkas' Lemma (see \cite[Corollary 7.1d on p. 89]{Schrijver1986}) equivalent to
	\begin{align*}
	\iff \quad & \exists \yb \in \R^{\M}_{\geq 0}, y^*\geq 0.\, \,\; (\yb, y^*)  \begin{pmatrix}
	- \Ab & \bb\\
	0 \ldots 0 & 1
	\end{pmatrix}= \begin{pmatrix}
	-\delta_{s_0} & \lambda
	\end{pmatrix} \\
	\iff \quad & \exists \yb \in \R^{\M}_{\geq 0}.\, \,\; \yb\Ab =\delta_{s_0}\land \yb\bb \leq \lambda	
	\label{eq:with equality}
	\end{align*}
	This shows in total the equivalence
	\begin{align*}
	& \exists \yb \in \R^{\M}_{\geq 0}.\, \,\; \yb\Ab =\delta_{s_0}\land \yb\bb \leq \lambda\\
	\iff \quad & \exists \yb \in \R^{\M}_{\geq 0}.\, \,\; \yb\Ab \geq\delta_{s_0}\land \yb\bb \leq \lambda
	\end{align*}
	and the analogous equivalence with strict inequalities on $\lambda$ follow similarly using a third version of Farkas' Lemma \cite[Corollary 7.1e on p. 89]{Schrijver1986}
	
	In \Cref{lem:relaxation on y} below we give a hands-on proof for this last equivalence which also provides the following interpretation of a non-negative vector $\yb$ with $\yb\Ab =\delta_{s_0}$: Write $\yb = (y_{s_0,\alpha_0}, \ldots, y_{s_0,\alpha_m},\ldots,y_{s_n,\alpha_m})$ and define the MR-scheduler $\S\colon S\to\Dist(\Act)$ by setting for $\alpha \in\Act$
	\[\S(s)(\alpha) = \frac{y_{s,\alpha}}{\sum_{\alpha\in \Act}y_{s,\alpha}}\]
	for those states where the denominator is positive, and $\S(s)$ arbitrary otherwise. This method for constructing MR-schedulers is standard, compare for example \cite[Chapter 2]{HordijkK84} or \cite[Theorem 3.2]{EtessamiKVY08}.
	Then $y_{s,\alpha}$ is the \emph{expected frequency} of the pair $(s,\alpha)$ occuring in a random walk on $\S$-paths in $\M$.
	For the induced DTMC $\M^\S$ it holds that $\Pr^{\S}_{s_0}(\lozenge\goal) = \yb \bb$, as $\bb$ contains the probability to move to $\goal$ in one step for every state and $\goal$ is absorbing.
	Hence
	\[\prb^{\min}_{s_0}(\lozenge\goal) \leq \Pr^{\S}_{s_0}(\lozenge\goal)= \yb \bb \leq \lambda\]
	which is the property we considered.
\end{remark}

\begin{lemma}\label{lem:relaxation on y}
		For the matrix $\Ab \in \R^{\M \times S}$ and vector $\bb\in \R^{\M}$ as in Setting \ref{setting:mdp}, we have 
	\begin{align*}
	\quad & \exists \yb \in \R^{\M}_{\geq 0}.\, \,\; \yb\Ab \leq\delta_{s_0}\land \yb\bb > \lambda \\
	\implies \quad & \exists \yb \in \R^{\M}_{\geq 0}.\, \,\; \yb\Ab =\delta_{s_0}\land \yb\bb > \lambda
	\end{align*}
\end{lemma}
\begin{proof}
		Let $\yb \geq 0$ satisfy $\yb\Ab \leq\delta_{s_0}\land \yb\bb > \lambda$. We define an MR-scheduler $\S$ of $\M$ by
		\[ \S(s,\alpha) = \frac{\yb(s,\alpha)}{\sum_{\beta \in \Act(s)}\yb(s,\beta)}\]
		for those states $s\in S$ for which the denominator is positive. In case that $\yb(s,\beta) = 0$ for all $\beta\in\Act(s)$, then we take an arbitrary $\S(s, \cdot) \in\Dist(\Act(s))$. Let $\Qb$ be the transition matrix of the DTMC $\M^\S$ induced by $\S$ (restricted to $S$), that is,
		\[\Qb(s,s') = \sum_{\alpha \in \Act(s)} \Pb(s,\alpha,s') \cdot \S(s,\alpha).\]
		Since $\Pr^{\S}_{s}(\lozenge (\goal\lor\fail))) = 1$ one sees as in the proof of \Cref{prop:LP version of pr_min} that $\Qb^n \to 0$ as $n\to\infty$. By invoking the Jordan normal form of $\Qb$ one deduces from this that $\Qb$ cannot have (complex) eigenvalues of absolute value greater than or equal to $1$. In turn, this implies that the series $\sum_{n \geq 0}  \Qb^n$ converges and that the limit is the inverse of $\Ib-\Qb$.
		
		Let $\hb$ be the expected frequencies of $\M^\S$ under initial distribution $\delta_{s_0}$, i.e., the solution of 
		\[\hb (\Ib - \Qb) = \delta_{s_0},\]
		where $\Ib\in\R^{S\times S}$ is the identity matrix.
		Now let us define 
		\begin{equation}\label{eq:y'}
			\yb'(s,\alpha) = \hb(s) \cdot \S(s,\alpha).
		\end{equation}
		Then it follows that for all $s\in S$ 
		\begin{equation}
			\hb(s) = \sum_{\alpha \in \Act(s)} \yb'(s,\alpha)\label{eqn:hbsumsup}
		\end{equation}
		and for those $s\in S$ with at least one positive $\yb(s,\alpha)$
		\begin{equation}\label{eq:mean equality}
		\frac{\yb'(s,\alpha)}{\sum_{\alpha \in \Act(s)}\yb'(s,\alpha)} = \frac{\yb(s,\alpha)}{\sum_{\alpha \in \Act(s)}\yb(s,\alpha)}.
		\end{equation}
		We now show that $\yb'$ satisfies $\yb'\Ab = \delta_{s_0}$.
		By the definition of $\hb$, we have
		\begin{align*}
		  \hb(s) &= \sum_{t \in S} \Qb(t,s) \cdot \hb(t) + \delta_{s_0}(s) \\
		  &= \sum_{t \in S} \sum_{\alpha \in \Act(t)} \Pb(t,\alpha,s) \cdot \S(t,\alpha) \cdot \hb(t) + \delta_{s_0}(s)
		\end{align*}
		By applying (\ref{eqn:hbsumsup}) on the left-hand side and (\ref{eq:y'}) on the right, we get
		\[\sum_{\alpha \in \Act(s)} \yb'(s,\alpha)
		= \sum_{(t,\alpha) \in \M}\Pb(t,\alpha,s)\cdot\yb'(t,\alpha) + \delta_{s_0}(s),\]
		which in total is precisely the desired equation of $\yb'\Ab = \delta_{s_0}$.

		Finally we show that $\yb'\bb \geq \yb\bb$ by proving the stronger statement that $\yb' \geq \yb$. We first claim for $\gb(s) = \sum_{\alpha \in \Act(s)}\yb(s,\alpha)$ that $\gb (\Ib-\Qb) \leq \delta_{s_0}$, i.e., 
		\[\gb(s) \leq  \sum_{t\in S} \gb(t) \cdot\Qb(t, s) + \delta_{s_0}(s)\]
		This is verified by the computation
		\begin{align*}
		\sum_{\alpha \in \Act(s)} \yb(s,\alpha) &\leq  \sum_{\substack{t\in S\\\gb(t)>0}} \left(\sum_{\alpha \in \Act(t)} \yb(t,\alpha)\right) \cdot 
												\left( \sum_{\beta\in\Act(t)} \Pb(t, \beta, s)\cdot\S(t,\beta)  \right)+ \delta_{s_0}(s)\\
						& =  \sum_{\substack{t\in S\\\gb(t)>0}} \left(\sum_{\alpha \in \Act(t)} \yb(t,\alpha)\right) \cdot 
						\left( \sum_{\beta\in\Act(t)} \Pb(t, \beta, s)\cdot \left( \frac{ \yb(t,\beta)}{\sum_{\gamma\in\Act(t)} \yb(t, \gamma) }  \right)\right)+ \delta_{s_0}(s)\\
						&=  \sum_{t\in S}  \sum_{\beta\in\Act(t)} \Pb(t, \beta, s)\cdot \yb(t,\beta)+ \delta_{s_0}(s),
		\end{align*}
		which is precisely the assumption $\yb\Ab \leq\delta_{s_0}$.
		
		Now since $\gb (\Ib-\Qb) \leq \delta_{s_0} = \hb (\Ib-\Qb)$, we have after multiplying on the right with $(\Ib-\Qb)^{-1} = \sum_{n\geq 0} \Qb^n$ that $\gb\leq \hb$, i.e., 
		\begin{equation} 
			  \sum_{\alpha \in \Act(s)}\yb(s,\alpha) \leq \sum_{\alpha \in \Act(s)}\yb'(s,\alpha)
		\end{equation}
		Because of (\ref{eq:mean equality}), this is equivalent to $\yb(s,\alpha)\leq \yb'(s,\alpha)$ for all $(s, \alpha)\in \M$ with $\gb(s) > 0$. For those states with $\gb(s) = 0$, there is nothing to prove for $\yb' \geq \yb$.
\qed\end{proof}

\section{Proofs of \Cref{sec:mws}}

\schedulerssubsys*
\begin{proof}
  As $\Act_{\M'}(s) = \Act_{\M}(s)$ for every $s \in S$, every choice of an MR-scheduler $\S$ in $\M$ induces an MR-scheduler $\S'$ of $\M'$ by restriction.
  To see that
  \[ \Pr^{\S'}_{\M', s_0}(\lozenge\goal) \leq \Pr^\S_{\M, s_0}(\lozenge\goal)\]
  holds for all MR-schedulers $\S$ of $\M$, it suffices to observe that 
  \[\{ \pi \in \Paths(\M') \mid \pi \models \lozenge \goal \} \subseteq \{ \pi \in \Paths(\M) \mid \pi \models \lozenge \goal \}\] 
  and that the probabilities of the paths in these sets is the same in $\M$ and $\M'$.
  As minimal and maximal reachability probabilities are attained by an MR-scheduler, we get as a consequence that $\prb^{\min}_{\M', s_0}(\lozenge\goal) \leq \prb^{\min}_{\M, s_0}(\lozenge\goal)$.
  
  For the $\prb^{\max}_{\M', s_0}(\lozenge\goal) \leq \prb^{\max}_{\M, s_0}(\lozenge\goal)$ one notices that, vice versa, every MR-scheduler $\S'$ on $\M'$ can be extended (arbitrarily) to an MR-scheduler $\S$ on $\M$ and that they satisfy $\Pr^{\S'}_{\M', s_0}(\lozenge\goal) \leq \Pr^{\S}_{\M, s_0}(\lozenge\goal)$. The rest of the argument is identical.
\qed\end{proof}

In the following lemma the \emph{size} of an MDP refers to sum of the number of states and the number of transitions, i.e., triples $(s,\alpha,t)$ with $\Pb(s,\alpha,t) >0$.

\begin{lemma}[Reduction to state-minimality]\label{lem:minreductions}
	\label{lem:minimal-reduction}
	Let $\M= (S_\all,s_0,\Act, \Pb)$ be an MDP as in Setting \ref{setting:mdp}. Then there exists an MDP $\mathcal{N}= (S'_\all,s_0,\Act, \Pb')$ such that the transition-minimal (respectively, size-minimal) witnesses of $\M$ are in one-to-one correspondence with the state-minimal witnesses of $\mathcal{N}$. The size of $\mathcal{N}$ is linear (respectively, quadratic) in the size of $\mathcal{M}$.
\end{lemma}
\begin{proof}
	Throughout this proof, we let $T$ denote the transitions of $\M$, i.e., the set of triples with $\Pb(s,\alpha,t) >0$. For the reduction from size-minimality to state-minimality, let $\mathcal{N}$ be the MDP with states $S_\all' = S_\all\cup T$ and transitions
	\begin{align*} 
		s\overset{\alpha}{\longrightarrow} (s, \alpha, t) \;&\text{ with probability }\Pb(s,\alpha,t)\\
		(s, \alpha, t) \overset{\alpha}{\longrightarrow} t\;&\text{ with probability }1
	\end{align*}
	Then there is a bijection between paths in $\M$ and paths in $\mathcal{N}$ given by
	\[ s_0\alpha_0s_1\alpha_1s_2... \;\text{ corresponds to } s_0\alpha_0(s_0,\alpha_0,s_1)\alpha_0 s_1\alpha_1(s_1,\alpha_1,s_2)\alpha_1 s_2...  \]
	and this bijection preserves probabilities. This immediately implies that a subsystem of $\M$ obtained by deleting states $S_d$ and transitions $T_d$ is a (minimal) witness if and only if the subsystem of $\mathcal{N}$ obtained by deleting states the $S_d\cup T_d$ is a (minimal) witness. Clearly, the reduction is linear in size. This finishes the reduction from size-minimality to state-minimality.
	
	For the reduction from transition-minimality, let $\mathcal{N}$ be the MDP with states $S_\all' = T\cup \{s_0,\goal,\fail\}$ and transitions
	\begin{align*} 
	s_0\overset{\alpha}{\longrightarrow} (s_0, \alpha, t) \;&\text{ with probability }\Pb(s_0,\alpha,t)\\
	(s, \alpha, t) \overset{\beta}{\longrightarrow} (t, \beta, u)\;&\text{ with probability } \Pb(t, \beta, u)\\
	(s, \alpha, \goal)\overset{\alpha}{\longrightarrow} \goal\;&\text{ with probability }1\\
	(s, \alpha, \fail)\overset{\alpha}{\longrightarrow} \fail\;&\text{ with probability }1
	\end{align*}
	Then there is again a probability-preserving bijection between paths in $\M$ and paths in $\mathcal{N}$. The rest of the argument is completely analogous. In this case, however, the size of $\mathcal{N}$ is quadratic in the size of $\M$ since there are $\mathcal{O}(|T|^2)$ many transitions of the second type in the above list.
\qed\end{proof}

\mwanpcompleteness*
\begin{proof}
	Assume that we are given an instance of the clique problem, i.e., a finite undirected graph $G=(V,E)$ and an integer $k\geq 3$ (the cases $k<3$ are trivial). Let $n=|V|$. Consider the DTMC $\M$ with states $S = \{s_0\}\cup V\cup E \cup \{\fail, \goal\}$ and four types of edges, see also \Cref{fig:sp graph}:
	\begin{itemize}
		\item $s_0\to v$ for every $v\in V$ with probability $1/n$;
		\item $v\to \{v,w\}$ for every $v\in V$ and edge $\{v,w\}\in E$ with probability $1/n$;
		\item $\{v,w\}\to \goal$ for edge $\{v,w\}\in E$ with probability $1$;
		\item $s\to \fail$ for all $s\in S$ with remaining probability outgoing from $s$ (provided there is some).
	\end{itemize}
	
	\begin{figure}[h]
		\centering
		\captionsetup{width=.8\linewidth, labelfont={bf,up}}
		
		\begin{tikzpicture}[->,>=stealth',shorten >=1pt,auto,node distance=0.5cm, semithick]
		
		\node[scale=0.8, state] (In) {$s_0$};
		\node[scale=0.8, state] (v) [below left = 1.5 and 1.5 of In] {$v$};
		\node[scale=0.8, state] (w) [below left = 1.5 and 0 of In]{$w$};
		\node[scale=0.8] (vdots) [below right = 1.5 and 0.5 of In]{$\cdots$};
		
		\node[scale=0.8, state] (vw) [below left = 1.5 and 0 of v]{$\{v,w\}$};
		\node[scale=0.8, state] (vu) [below right= 1.5 and 1 of v]{$\,\{v,u\}\,$};
		\node[scale=0.8] (edots) [below right = 1.5 and 2.9 of v]{$...$};
		\node[scale=0.8, state] (g) [below = 5.5 of In]{$\goal$};
		
		\draw (In) -- (v) node[pos=0.5,scale=0.8, xshift=-8mm, yshift=2mm] {$\frac{1}{n}$};
		\draw (In) -- (w) node[pos=0.5,scale=0.8,  xshift=-5mm, yshift=-1mm] {$\frac{1}{n}$};
		\draw (In) -- (vdots) node[pos=0.5,scale=0.8,  xshift=1mm, yshift=-4mm] {$\frac{1}{n}$};
		
		\draw (v) -- (vw) node[pos=0.3,scale=0.8, xshift=-6mm] {$\frac{1}{n}$};
		\draw (w) -- (vw) node[pos=0.3,scale=0.8,  xshift=-10mm, yshift=-0.5mm] {$\frac{1}{n}$};
		\draw (v) -- (vu) node[pos=0.7,scale=0.8, xshift=-0mm, yshift=-2mm] {$\frac{1}{n}$};
		\draw (vw) -- (g) node[pos=0.5,scale=0.8] {$1$};
		\draw (vu) -- (g) node[pos=0.5,scale=0.8] {$1$};
		
		\draw (w) -- (edots) node[pos=0.5,scale=0.8] {};
		
		\end{tikzpicture}
		\caption{The acyclic DTMC $\M$ reducing the clique problem to the witness problem. The state $\fail$ and edges to it are not depicted for simplicity.}\label{fig:sp graph}
	\end{figure}
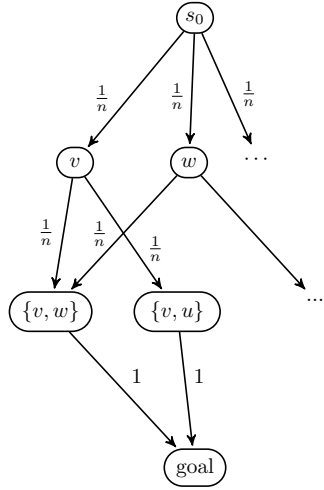
	We claim that the graph $G$ has a clique with at least $k$ vertices if and only if the MDP $\M$ has a witness for $\Pr_{\M, s_0}(\lozenge\goal) \geq \lambda := \frac{k(k-1)}{n^2}$ with at most $k' :=k+ \frac{k(k-1)}{2} + 3$ states (the $3$ accounting for $s_0, \fail$, and $\goal$). The 'only if' direction is clear. 
	
	For the reverse direction assume that $\M$ has a witness $\M'$ for $\Pr_{\M, s_0}(\lozenge\goal) \geq \frac{k(k-1)}{n^2}$ with at most $k' :=k+ \frac{k(k-1)}{2} + 3$ states. Denote those states of $\M'$ lying in $V$ by $V'$ and those lying in $E$ by $E'$. Let $a = |V'|$ and $b=|E'|$, and recall that $a+b \leq k+\frac{k(k-1)}{2}$. We intend to show that $V'$ is a $k$-clique. 
	Since $\Pr_{\M', s_0}(\lozenge\goal) \geq  \frac{k(k-1)}{n^2}$, we have at least $k(k-1)$ transitions between states in $V'$ and $E'$. Since each state in $E'$ has exactly two incoming transitions, we have $k(k-1)\leq 2b$, and therefore $a\leq k$.
	At most $\frac{a(a-1)}{2}$ states of $E'$ have two incoming transitions from $V'$, the others have at most one incoming transition from $V'$. Thus the total number of transitions from $V'$ to $E'$ is bounded from above by
	\begin{align*}
		& \quad\;2\cdot \frac{a(a-1)}{2} + b - \frac{a(a-1)}{2} \\
		&\leq	2\cdot \frac{a(a-1)}{2} +  \left(k+\frac{k(k-1)}{2}-a\right)- \frac{a(a-1)}{2} \\
				&=  \frac{a(a-3)}{2} + \frac{k(k+1)}{2} - k(k-1) + k(k-1)\\
				&=  \frac{a(a-3)}{2} - \frac{k(k-3)}{2} + k(k-1)\\
				& \leq k(k-1)
	\end{align*}
	where the last step follows from $a\leq k$ and $k\geq 3$. But by assumption, the total number of transitions from $V'$ to $E'$ is at least $k(k-1)$, and hence we have everywhere equality in the above computation. This can only happen if $a =k$, $b = \frac{k(k-1)}{2}$ and every edge in $E'$ has two incoming transitions from $V'$. This forces $V'$ to be a $k$-clique with edge set $E'$.
	\qed\end{proof}

\section{Proofs of \Cref{sec:relation}}

\polytope*

\begin{proof}
	This proof has some resemblence with the proof \Cref{prop:LP version of pr_min}. If $\P^{\min}(\lambda)$ was unbounded, then there exists $\zb_0, \zb_1\in \R^S$ with $\zb_1\neq 0$ such that $\zb_0+t\zb_1\in \P_\M^{\min}(\lambda)$ for all $t\in [0,\infty)$. This implies $\Ab\zb_1\leq 0$ and $\zb_1(s_0)\geq 0$. Let $m = \max\{\zb_1(s)\mid s\in S\}$ and let $S_{\max} = \{s\in S\mid \zb_1(s) = m\}$.
	Assume that $m=0$.
	Then $s_0\in S_{\max}$
	\[0= \zb_1(s_0) \leq \sum_{t\in S} \Pb(s_0,\alpha,t)\cdot\zb_1(t)\leq \sum_{t\in S} \Pb(s_0,\alpha,t)\cdot m =0\]
	and hence everywhere equality. This implies that $\zb_1(t) = 0$ for all successors of $s_0$ and by induction $\zb_1(t) = 0$ for all reachable states $t$. By assumption all states are reachable from $s_0$ and hence $\zb_1 = 0$, a contradiction.
	
	Now fix some arbitrary $s_{\max}\in S_{\max}$. Then again because of $\Ab\zb_1\leq 0$ we have for all $\alpha\in \Act(s)$	
	\[m= \zb_1(s_{\max}) \leq \sum_{t\in S} \Pb(s_{\max},\alpha,t)\cdot\zb_1(t)\leq \sum_{t\in S} \Pb(s_{\max},\alpha,t)\cdot m \leq m \]
	and hence everywhere equality. But then $\sum_{t\in S} \Pb(s_{\max},\alpha,t) = 1$, meaning that there is no transition from $s_{\max}$ to $\fail$ or $\goal$, and $\zb_1(t) = m$, meaning that the same applies to all successors of $s_{\max}$. This implies by induction that $\goal$ and $\fail$ are not reachable from $s_{\max}$, a contradiction to the assumption $\prb_{s_{\max}}^{\min}(\lozenge (\goal\lor \fail)) >0 $.
	
	\medskip
	The argument for $\P^{\max}(\lambda)$ is similar: Assume that this polyhedron was unbounded, i.e., there exist $\yb_0, \yb_1\in\R^{\M}$ with $\yb_1\neq 0$ such that  $\yb_0+t\yb_1\in\P^{\max}(\lambda)$  for all $t\in [0,\infty)$. Then necessarily $\yb_1\geq 0$ and $\yb_1\Ab\leq 0$, and thus also $\yb_1\Ab (1,..., 1)\leq 0$. On the other hand, $\Ab (1,..., 1) \geq 0$ and therefore $\yb_1\Ab (1,..., 1)\geq 0$. So $\yb_1\Ab (1,..., 1)= 0$ which means $\yb_1\Ab = 0$. Thus, if $\yb_1(s, \alpha) > 0$, then $\Ab(1,...,1) (s, \alpha) = 0$, i.e., there does not exist a transition $s\overset{\alpha}{\longrightarrow} \goal$ or $s\overset{\alpha}{\longrightarrow} \fail$. Now let 
	\begin{gather*}
	S_+ = \{ s\in S\mid \sum_{\alpha\in\Act(s)} \yb_1(s,\alpha) >0 \} \\
	S_0 = \{ s\in S\mid \sum_{\alpha\in\Act(s)} \yb_1(s,\alpha) =0  \}
	\end{gather*}
	Since $\yb_1\neq 0$, we have $S_+ \neq\varnothing$, so we can fix $s_+\in S_+$. If $s\overset{\alpha}{\longrightarrow} t$ is a transition in $\M$ with $\yb_1(s, \alpha) >0$, then because of $\yb_1\Ab= 0$ we also have $t\in S_+$. So if we let $\S\colon S\to\Act$ be the MR-scheduler defined on $s\in S_+$ by
	\[ \S(s,\alpha) = \frac{\yb(s,\alpha)}{\sum_{\beta \in \Act(s)}\yb(s,\beta)}\]
	and on $S_0$ arbitrarily, then $\Pr_{s_+}^{\S}(\lozenge (\goal\lor \fail)) = 0$. However, we assumed that $\prb_{s_+}^{\min}(\lozenge (\goal\lor \fail))>0$.
\qed\end{proof}

\mainthm*

\begin{proof}
	(1)$\implies$(2): Assume that $\M_R$ is a witness for $\prb^{\min}_{\M, s_0}(\lozenge\goal)\geq\lambda$, i.e., $\prb^{\min}_{\M_R, s_0}(\lozenge\goal) \geq\lambda$. Let $\pb\in\R^S$ be the point containing a zero for the states $s\notin R$ and $\prb^{\min}_{\M_R, s}(\lozenge\goal)$ otherwise. Then $\pb\geq 0$, $\supp(\pb) \subseteq R$, and by assumption $\pb(s_0) \geq\lambda$. It remains to show that $\Ab\zb\leq \bb$, i.e. that for all $s\in S$ and $\alpha\in \Act(s)$ we have
	\begin{equation}\label{eq:Az b} 
	\pb(s) \leq \sum_{t\in S} \Pb(s,\alpha, t) \cdot\pb(t) + \bb(s, \alpha).
	\end{equation}
	For those states with $s\notin R$ or those with $\prb^{\min}_{\M_R, s}(\lozenge\goal)=0$, this  is clear since the left-hand side vanishes and the right-hand side is non-negative. The other states form a subset $R'\subseteq R$ and the vector $\pb'\in\R^{R'}$ obtained by projection from $\pb$ is equal to $\pb' = (\prb^{\min}_{\M_{R'}, s}(\lozenge\goal))_{s\in R'}$. Therefore we clearly have for all $s\in R'$ and $\alpha\in \Act(s)$
	\[ \pb'(s) \leq  \sum_{t\in R'} \Pb(s,\alpha, t) \cdot\pb'(t)+ \bb(s, \alpha)\]
	This implies (\ref{eq:Az b}) since the right-hand side in both inequalities agree. Hence $\pb\in\P^{\min}_{\geq 0}(\lambda)$.
	
	\medskip
	(2)$\implies$(1): Take a point $\pb\in\P^{\min}_{\geq 0}(\lambda)$ with $\supp(\pb) \subseteq  R$. Since 
	\[\prb^{\min}_{\M_R, s_0}(\lozenge\goal)\geq \prb^{\min}_{\M_{\supp(\pb)}, s_0}(\lozenge\goal) \]
	we may assume that $\supp(\pb) =  R$. Since $\pb(s_0) \geq\lambda$, it suffices to prove
	\[\prb^{\min}_{\M_R, s_0}(\lozenge\goal)\geq \pb(s_0)\]
	which we achieve by invoking \Cref{lem:monotonicity}.	We may assume that every state in $R$ is reachable from $s_0$ in $\M_R$, otherwise we restrict to those states. Notice also that for all $s\in R$ we have
	\begin{equation*}\label{eq:2-1}
	\prb^{\min}_{\M_R, s}(\lozenge (\goal\lor\fail)) = \prb^{\min}_{\M, s}(\lozenge (\goal\lor\fail)) > 0
	\end{equation*} 
	Now let $\pb'$ be the projection of $\pb\in \R^S$ to $\R^R$ and let $\Ab'\in\R^{\M_R\times R}$ and $\bb'\in \R^{\M_R}$ be defined similarly. Then $\Ab'$ and $\bb'$ are precisely as in Setting \ref{setting:mdp} for the MDP $\M_R$. By assumption we have $\Ab\pb\leq\bb$ and since $\supp(\pb)\subseteq R$ we also have $\Ab'\pb'\leq\bb'$. We now apply \Cref{lem:monotonicity} and get $\pb'\leq \prb^{\min}_{\M_R}(\lozenge\goal)$, and therefore $\prb^{\min}_{\M_R}(\lozenge\goal)\geq \pb'(s_0) = \pb(s_0)\geq\lambda$.

	\medskip
	(2)$\implies$(3): This is a general observation from polytope theory: Take a point $\pb$ in $\P = \P^{\min}_{\geq 0}(\lambda)$ such that $\supp(\pb)\subseteq R$. Let 
	\[\mathcal{H} = \{\xb\in \R^{\M}\mid \xb(s, \alpha) = 0\;\text{ for all } (s,\alpha)\in \M\setminus \supp(\pb)\}.\] 
	As the inequalities $\xb(s, \alpha)\geq 0$ are also part of the description of $\P$ and $\pb\in \P\cap\mathcal{H}$, the set $\P\cap\mathcal{H}$ is a face of $\P$. This face has a vertex $\vb\in \P\cap \mathcal{H}$. Then $\vb$ is also a vertex of $\P$ and we have $\supp(\vb)\subseteq \supp(\pb)  \subseteq R$.
	
	\medskip
	(3)$\implies$(2): Trivial.
	
	\medskip	
	(a)$\implies$(b): Let $\M_R=(S'_\all, s_0, \Act, \Pb')$ and denote $S' = S'_\all \setminus\{\goal, \fail\}$. If $\M_R$ is a witness for $\prb^{\max}_{\M, s_0} (\lozenge\goal)\geq\lambda$, then there exists an MD-scheduler $\S: S'_\all\to \Act$ such that $\Pr_{\M_R, s_0}^{\S}(\lozenge\goal)\geq \lambda$. Let $\Qb\in\R^{S'\times S'}$ be the transition matrix of $\M_R^{\S}$, and likewise let $\cb\in\R^{S'}$ contain the probabilities to go from $s$ to $\goal$ in one step in $\M_R^{\S}$. 
	
	Recall that for all $s\in R$ we have $\prb^{\min}_{\M_R, s}(\lozenge (\goal\lor\fail)) = \prb^{\min}_{\M, s}(\lozenge (\goal\lor\fail)) > 0 $.
	With the same argument as in the proof of \Cref{prop:LP version of pr_min} one sees that $\Qb^n\to 0 $ as $n\to\infty$. By invoking the Jordan normal form of $\Qb$ one deduces from this that $\Qb$ cannot have (complex) eigenvalues of absolute value greater than or equal to $1$. In turn, this implies that the series $\sum_{n \geq 0}  \Qb^n$ converges.
	
	Now let $\qb = \delta_{s_0}\cdot \sum_{n \geq 0}  \Qb^n\in\R^{S'}$. Recall that $\Qb^n(s_0, s)$ is the probability that a path of length $n$ starting in $s_0$ ends in $s$. From this it is easy to see that $\qb\cb =  \Pr_{\M_R, s_0}^{\S}(\lozenge\goal)$. Now let $\pb\in \R^{\M}$ be the vector with $\pb(s, \alpha) = \qb(s)$ if $s\in S'$ and $\S(s) = \alpha$, and $0$ otherwise. Then clearly $\supp(\pb)\subseteq R$ and we claim that $\pb\in \P^{\max}(\lambda)$. Obviously, we have $\pb \geq 0$ and also $\pb \bb = \qb\cb = \Pr_{\M_R, s_0}^{\S}(\lozenge\goal) \geq \lambda$. Furthermore,
	\[ \qb(\Ib-\Qb) =  \delta_{s_0}\cdot \sum_{n \geq 0}  \Qb^n \cdot (\Ib-\Qb) = \delta_{s_0}\]
	This implies $\pb\Ab \leq \delta_{s_0}$ since for the states $s\in S\setminus S'$ there is nothing to show in this inequality.
	
	\medskip	
	(b)$\implies$(a): Let $\pb$ be an element of $\P^{\max}(\lambda)$, let $U = \supp(\pb)\subseteq R$. We will show that $\M_U= (S'_\all, s_0, \Act, \Pb')$ is a witness for $\prb^{\max}_{\M, s_0} (\lozenge\goal)\geq\lambda$, which immediately implies the same for $\M_R$. Let $\Ab'$ and $\bb'$ be as in Setting \ref{setting:mdp} for $\M_U$, and denote $S' = S'_\all\setminus\{\goal,\fail\}$. Let $\pb'$ be the projection of $\pb\in \R^{\M}$ to $\R^{\M_U}$.
	
	Then 
	\begin{equation}\label{eq:pb's}
		\pb'\cdot\bb' = \pb\cdot\bb\geq \lambda
	\end{equation}
	since $\bb'(s, \alpha) = \bb(s, \alpha)$ for all $(s, \alpha)\in \supp(\pb) = U$. Note also that for $s\in S'$ we get from $\pb\Ab\leq\delta_{s_0}$ 
	\begin{align*}
		\sum_{(s, \alpha)\in U} \pb'(s, \alpha)
		 & = \sum_{\alpha\in\Act(s)} \pb(s, \alpha)\\ 
		 &\leq \sum_{t\in S}\sum_{\beta\in \Act(t)} \Pb(t, \beta, s)\cdot \pb(t, \beta) +\delta_{s_0}(s) \\
		 &= \sum_{t\in S'}\sum_{(t,\beta)\in U} \Pb'(t, \beta, s)\cdot \pb'(t, \beta) +\delta_{s_0}(s) 
	\end{align*}
	which can concisely be written as
	\begin{equation}\label{eq:Ap's}
		\pb'\cdot \Ab'\leq \delta_{s_0}
	\end{equation}

	Take as in the proof of \Cref{lem:relaxation on y} the MR-scheduler $\S$ on $\M_U$ defined by
	\[ \S(s,\alpha) = \frac{\pb(s,\alpha)}{\sum_{\beta \in \Act(s)}\pb(s,\beta)}\]
	We consider the DTMC $\M_U^{\S}$ and we will show that the probability to reach $\goal$ in this DTMC is greater than $\lambda$ in order to show that $\M_U$ is a witness. As before denote the transition matrix of this DTMC $\Qb\in\R^{S'\times S'}$ and collect the probabilities to go from a state to $\goal$ in one step in $\cb\in \R^{S'}$. As in the proof of (a)$\implies$(b) we see that the matrix series $ \sum_{n \geq 0}  \Qb^n$ converges. Also
	\begin{equation}\label{eq QCP}
	\Qb = \Cb \cdot \Pb',
	\end{equation}
	where $\Cb\in\R^{S'\times\M_U}$ has entries 
	\[
	\Cb(s, (t, \alpha)) = \begin{cases} 
	\S(s, \alpha) & \mbox{if } s=t \\
	0  & \mbox{otherwise } \\\end{cases}
	\]
	Likewise we have
	\[  \cb = \Cb\cdot \bb'\]
	The matrix $\Cb$ has a left-inverse $\Bb\in\R^{\M_U\times S'}$ with entries
	\[
	\Bb((s, \alpha), t) = \begin{cases} 
	1/\S(s, \alpha) & \mbox{if } s=t \\
	0  & \mbox{otherwise } \\\end{cases}
	\]
	satisfying
	\begin{equation}\label{eq:Bcb'}
		\Bb\cdot\cb= \bb'
	\end{equation}
	If we let $\mathbf{J}\in\R^{\M_U\times S'}$ be the matrix with entries $\mathbf{J}((s,\alpha), t) = \delta_{st}$, then we have the relation $\Ab' = \Jb-\Pb'$. Notice also that $\Cb\cdot \Jb = \Ib_{S'}$. This implies with (\ref{eq QCP}) that
	\begin{equation*}
	\Cb\cdot\Jb\cdot \Qb^{k+1} =\Ib_{S'}\cdot \Qb^{k+1} = \Qb^{k+1} = \Cb\cdot \Pb'\cdot \Qb^k
	\end{equation*}
	and therefore we have the following telescope sum
	\begin{equation}\label{eq:telescope}
	\Cb\cdot\Ab'\cdot\left(\sum_{k\geq 0} \Qb^{k}\right) = \Cb\cdot(\Jb-\Pb')\cdot\left(\sum_{k\geq 0} \Qb^{k}\right) = \Ib_{S'}.
	\end{equation}
	
	Putting everything together allows us to calculate
	\begin{align*}
	\Pr_{\M_U, s_0}^{\S}(\lozenge\goal)&\overset{}{=} \delta_{s_0}\cdot \left( \sum_{k\geq 0} {\Qb}^k\right)\cdot \cb\\
	&\overset{\text{(\ref{eq:Ap's})}}{\geq} \pb'\cdot\Ab' \cdot \left( \sum_{k\geq 0} \Qb^k\right)\cdot \cb\\
	&=\pb'\cdot\Bb\cdot \Cb \cdot\Ab' \cdot \left(\sum_{k \geq 0}\Qb^k\cdot  \right)\cdot \cb\\
	&\overset{\text{(\ref{eq:telescope})}}=\pb'\cdot\Bb\cdot \Ib_{S'}\cdot\cb\\
	&\overset{\text{(\ref{eq:Bcb'})}}=\pb'\cdot\bb'\\
	&\overset{\text{(\ref{eq:pb's})}}\geq \lambda.
	\end{align*}
	
	\medskip
	(b)$\;\Longleftrightarrow\;$(c): This is the same argument as for (2)$\;\Longleftrightarrow\;$(3).
\qed\end{proof}

\verticestoMW*
\begin{proof}
	We prove the statement about $\P = \P^{\max}(\lambda)$, the dual statement follows along similar lines. For $\vb\in \R^{\M}$ we set $\supp_S(\vb) = \{s\in S\mid \exists\alpha\in\Act.\; \vb(s, \alpha)> 0\}$ .
	
	We begin with the 'if' part, and assume that there exists a vertex $\mathbf{w}$ of $\P$ with a strictly larger number of zeros than $\vb$. Since there is for every $s\in S' = \supp_S(\vb)$ only one pair $(s, \alpha)$ in $\supp(\vb)$, this implies that $|\supp_S(\mathbf{w})| < |\supp_S(\vb)|$. But by \Cref{thm:MCS-polytope}, (c) $\Longrightarrow$ (a), the subsystem $\M_{\supp(\mathbf{w})}$ is also a witness, and it would contain a strictly smaller number of states than $\M_{\supp(\vb)}$. Contradiction to the minimality of $\M_{\supp(\vb)}$.
	
	For the 'only if' part, let $\vb$ be a vertex of $\P$ with a maximal number of zeros. Again by \Cref{thm:MCS-polytope}, (c) $\Longrightarrow$ (a), $\M_{\supp(\vb)}$ is a witness. If it was not minimal, then there is a set $R\subseteq \M$ such that $\M_R$ is a witness with a strictly smaller number of states than $\M_{\supp(\vb)}$. Now the proof of \Cref{thm:MCS-polytope}, (a) $\Longrightarrow$ (b) $\Longrightarrow$ (c) provides a vertex $\mathbf{w}$ of $\P$ with $\supp(\mathbf{w})\subseteq R$ and for every $s\in\supp(\mathbf{w})$ there is precisely one $\alpha\in\Act$ such that $(s,\alpha)\in \supp(\mathbf{w})$. This implies 
	\begin{align*}
	|\supp(\mathbf{w})| &= |\supp_S(\mathbf{w})| \\
	&\leq \text{number of states in } \M_R \\
	&< \text{number of states in } \M_{\supp(\vb)} \\
	& \leq |\supp(\vb)|,
	\end{align*}
	which is a contradiction. Hence $\M_{\supp(\vb)}$ is a minimal witness.
	
	Almost the same argument shows that for every $s\in S'$ there is precisely one $\alpha\in\Act$: Otherwise one can again invoke the proof of \Cref{thm:MCS-polytope}, (a) $\Longrightarrow$ (b) $\Longrightarrow$ (c) applied to $\M_{\supp(\vb)}$ and an MD-scheduler $\S'$ attaining $\prb^{\max}_{\M_{\supp(\vb)}, s_0}(\lozenge\goal)\geq \lambda$ in order to obtain a vertex of $\P$ with a greater number of zeros than $\vb$. 
	
	Finally, the inequality $\Pr^\S_{\M_{\supp(\vb)}, s_0} (\lozenge\goal)\geq \vb\bb \geq \lambda$ follows with the same arguments as in the proof of \Cref{thm:MCS-polytope}, (b)$\implies$(a), so $\S$ is indeed a witnessing scheduler for $\prb^{\max}_{s_0}(\lozenge\goal)\geq \lambda$.
\qed\end{proof}

\section{Proofs for \Cref{sec:computing}}

\MILPformulation*
\begin{proof}
  Suppose that the conditions are satisfied for $\mathcal{P} = \{ \xb \mid \Ab \xb \leq \bb, \xb \geq 0 \} \subseteq \mathbb{R}^n$ and $K \geq 0$.

  We first show that for any optimal solution $(\boldsymbol{\sigma},\xb)$ of the MILP, we have $\xb(i) = 0$ if and only if $\boldsymbol{\sigma}(i) = 0$.
  If $\xb(i) = 0$ and $\boldsymbol{\sigma}(i) = 1$, we would get a better solution by setting $\boldsymbol{\sigma}(i) = 0$, contradicting the fact that $(\boldsymbol{\sigma},\xb)$ is optimal.
  If $\boldsymbol{\sigma}(i) = 0$, then $\xb(i)$ must be zero as $\xb(i) \leq K \cdot \boldsymbol{\sigma}(i)$.
  We write $\boldsymbol{\sigma}(\xb)$ for the vector that has a $1$ at every position where $\xb$ is greater than $0$, and $0$ otherwise.

  For every point $\pb \in \mathcal{P}$, $(\boldsymbol{\sigma}(\pb),\pb)$ satisfies the constraints of the MILP.
  To see that $\pb \leq K \cdot \boldsymbol{\sigma}(\pb)$, we observe that for all $i: \, K \cdot \boldsymbol{\sigma}(\pb)(i) \leq K$, and $K$ was chosen exactly such that for all $i: \, \xb(i) \leq K$ for all $\xb \in \mathcal{P}$.
  Now we can show the claim:

  ``$\implies$'': Let $(\boldsymbol{\sigma},\xb)$ be an optimal solution of the MILP.
  Clearly, $\xb \in \mathcal{P}$.
  By the above argument, $\boldsymbol{\sigma} = \boldsymbol{\sigma}(\xb)$.
  Suppose that there is another point $\pb$ in $\mathcal{P}$ with more zeros.
  But then $(\boldsymbol{\sigma}(\pb),\pb)$ is a solution of the MILP and $\boldsymbol{1} \cdot \boldsymbol{\sigma}(\pb) < \boldsymbol{1} \cdot \boldsymbol{\sigma}(\xb)$.
  Hence $(\boldsymbol{\sigma},\xb)$ is not an optimal solution to the MILP, contradicting the assumption.
  
  ``$\Longleftarrow$'': Let $\xb$ be a point in $\P$ with a maximal number of zeros.
  Then $(\boldsymbol{\sigma}(\xb),\xb)$ satisfies the constraints of the MILP.
  Furthermore, it is an optimal solution, as a better solution would contradict the maximality of the number of zeros in $\xb$.
\end{proof}

\section{Polynomial algorithm in the tree-shaped case}

In this section we show that a minimal witness for the property $\Pr_{\M}(\lozenge \goal) \geq \lambda$ can be computed in polynomial time for tree-shaped DTMCs, given that it exists. Here \emph{tree-shaped} refers to the property that the underlying graph of $\M$ excluding $\goal$ and $\fail$, i.e. the graph with vertices $V=S$ and edges $E = \{(s,t) \in S \times S \mid \Pb(s,t) > 0\}$, is a tree.
The algorithm has two steps: first we reduce a tree-shaped DTMC to the special case of a binary tree-shaped DTMC. Then, we provide an algorithm for the binary case whose result can be translated back to the general case.

We consider, as in Setting~\ref{setting:mdp}, DTMCs with distinguished, absorbing states $\goal$ and $\fail$.
The predicates \emph{acyclic} and \emph{binary} refer to the underlying graph the DTMC excluding $\goal$ and $\fail$.

\subsubsection{Binarization of Markov chains.}

We first give a transformation of a Markov chain into a binary Markov chain that preserves the probability to reach $\goal$.

\begin{definition}[Binarization]
  \label{def:binarization}
	Let $\mathcal{M} = (S_{\all},s_0,\Pb)$ be a DTMC as in Setting~\ref{setting:mdp} and $<$ a total order on $S$ such that $\goal$ is its minimal element.
  For every state $q \in S$, let $\sucs(q)=\{s\in S_\all \mid \Pb(q,s) > 0\}$.
	We define $\mathcal{B}(\mathcal{M},<) = (S_{\all}',s_0,\Pb')$ by giving a local transformation that is applied to all states $q \in S_{\all}$ with $|\sucs(q)| > 2$.
  States with less than three successors stay as they are.

  Let $q \in S_{\all}$, $\sucs(q) = \{s_0,s_1,\ldots,s_n\}$ be ordered according to $<$ and define $\mu_i = \Pb(q,s_i)$ for $0 \leq i \leq n$.
  Take $n{-}1$ fresh states $u_1,\ldots,u_{n{-}1}$.
  The new transition probabilities are defined as follows, where we identify $q$ with $u_0$: 
  \begin{align*}
    &\Pb'(u_j,s_j) = \frac{\mu_j}{1{-}\sum_{0 \leq i < j}\mu_i} \;\;\text{, for}\;\; 0 \leq j < n \\
    &\Pb'(u_j,u_{j+1}) = 1{-}\Pb'(u_j,s_j) \;\;\text{, for}\;\; 0 \leq j < n{-}1 \\
    &\Pb'(u_{n{-}1},s_{n}) = \frac{\mu_{n}}{1{-}\sum_{0 \leq i < n{-}1}\mu_i}
  \end{align*}
\end{definition}

\begin{remark}
  The condition that $\goal$ is the minimal element of $<$ implies for every $q \in S$ that if $\goal \in \sucs(q) = \{s_0,s_1,\ldots,s_n\}$, then $\goal = s_0$.
  This makes sure that for every state $u \in U$ the probability to reach $\goal$ in one step is zero, i.e. $\Pb'(u,\goal) = 0$.
\end{remark}

The resulting Markov chain is binary by construction and its state space consists of the old states (with adapted outgoing transitions) and the states added by the construction. 
In what follows, we fix any total order $<$ such that $\goal$ is the minimal element of $<$ and write $\mathcal{B}(\mathcal{M}) = \mathcal{B}(\mathcal{M},<)$.
\Cref{fig:binarization} shows how the transformation works.

\begin{figure}[h]
	\captionsetup{width=0.8\linewidth, labelfont={bf, up}}
	\begin{subfigure}[]{0.5\columnwidth}
		\centering
		\scalebox{1.}{
		\begin{tikzpicture}[->,>=stealth',shorten >=1pt,auto,node distance=0.5cm, semithick]
		
		\node[scale=0.8, state] (In) {$q$};
		\node[scale=0.8, state] (s0) [below left = 1.5 and 1.8 of In] {$s_0$};
		\node[scale=0.8, state] (s1) [below left = 1.5 and 0.6 of In]{$s_1$};
		\node[scale=0.8, state] (s2) [below right = 1.5 and 0.6 of In]{$s_2$};
		\node[scale=0.8, state] (s3) [below right = 1.5 and 1.8 of In]{$s_3$};

		\draw (In) -- (s0) node[pos=0.5,scale=0.8, xshift=-8mm, yshift=2mm] {$0.4$};
		\draw (In) -- (s1) node[pos=0.5,scale=0.8,  xshift=-8mm, yshift=-2mm] {$0.3$};
		\draw (In) -- (s2) node[pos=0.5,scale=0.8,  xshift=1mm, yshift=-6.5mm] {$0.2$};
		\draw (In) -- (s3) node[pos=0.5,scale=0.8,  xshift=1mm, yshift=-3mm] {$0.1$};
		
		\end{tikzpicture}
	}
		\subcaption{}
		\label{fig:successors}
	\end{subfigure}
	\begin{subfigure}[]{0.5\columnwidth}
		\centering
		\scalebox{1.}{
			\begin{tikzpicture}[->,>=stealth',shorten >=1pt,auto,node distance=0.5cm, semithick]
			
			\node[scale=0.8, state] (In) {$q$};
			\node[scale=0.8, state] (s0) [below left = 0.7 and 0.5 of In] {$s_0$};
			\node[scale=0.8, state] (u1) [below right = 0.7and 0.5 of In]{$u_1$};
			\node[scale=0.8, state] (s1) [below left = 0.7 and 0.5 of u1]{$s_1$};
			\node[scale=0.8, state] (u2) [below right = 0.7 and 0.5 of u1]{$u_2$};
			\node[scale=0.8, state] (s2) [below left = 0.7 and 0.5 of u2]{$s_2$};
			\node[scale=0.8, state] (s3) [below right = 0.7 and 0.5 of u2]{$s_3$};
			
			\draw (In) -- (s0) node[pos=0.5,scale=0.8, xshift=-8mm, yshift=2.8mm] {$0.4$};
			\draw (In) -- (u1) node[pos=0.5,scale=0.8, xshift=-0mm, yshift=-2mm] {$0.6$};
			\draw (u1) -- (s1) node[pos=0.5,scale=0.8,  xshift=-7mm, yshift=3mm] {$0.5$};
			\draw (u1) -- (u2) node[pos=0.5,scale=0.8,  xshift=0mm, yshift=-1.5mm] {$0.5$};
			\draw (u2) -- (s2) node[pos=0.5,scale=0.8,  xshift=-8mm, yshift=3mm] {$2/3$};
			\draw (u2) -- (s3) node[pos=0.5,scale=0.8,  xshift=0mm, yshift=-3mm] {$1/3$};
			
			\end{tikzpicture}
			
		}
		\subcaption{}
		\label{fig:binary}
	\end{subfigure}
	\caption{Illustration of the local transformation of \Cref{def:binarization} from a state with four successors (\subref{fig:successors}) to its binarization (\subref{fig:binary}).}\label{fig:binarization}
\end{figure}
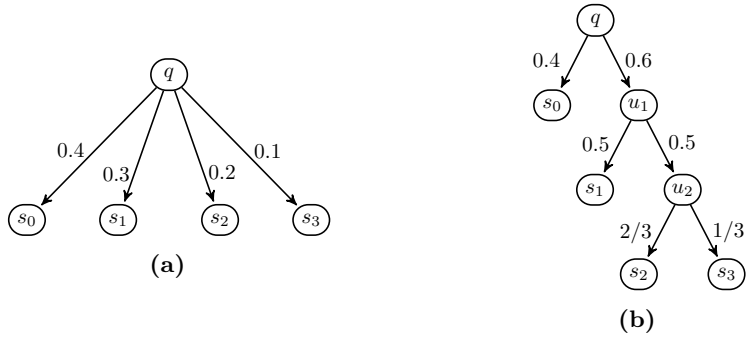

\begin{remark}
  If we start with a tree-shaped Markov chain, the result of the above construction is also tree-shaped.
  The number of states of $\mathcal{B}(\M)$ is bounded by the number of states plus the number of transitions of $\M$.
\end{remark}

\begin{lemma}
	\label{lem:propbm}
  There is a probability preserving one-to-one correspondence of paths in $\M = (S_{\all},s_0,\Pb)$ and paths in $\mathcal{B}(\M) = (S_{\all} \cup U,s_0,\Pb')$, where $U$ is the set of all fresh states added by the construction.
\end{lemma}
\begin{proof}
  From a path $\pi$ in $\mathcal{B}(\M)$ we get a path $\pi'$ in $\M$ by removing from $\pi$ all states in $U$.
  We show that $\Pr_{\mathcal{B}(\M)}(\pi) = \Pr_\M(\pi')$.
  Let $S = S_{\all} \setminus \{\goal,\fail\}$.
  It suffices to show this for all paths of the form $SU^+S$, so let $\pi = q u_1 \ldots u_n s$ with $q,s \in S$, $u_i \in U$ for $1 \leq i \leq n$ and $n \geq 1$.
  As in Definition~\ref{def:binarization}, let $\sucs(q) = \{s_0, s_1, \ldots s_m\}$ be ordered according to $<$, with $m > n$, $\mu_i = \Pb(q,s_i)$ for all $s_i \in \sucs(q)$ and $s = s_{l}$.
    By the way the binarization was defined, we get that $l \in \{n,n+1\}$.
  The construction now gives us:
  \begin{equation*}
    \Pr_{\mathcal{B}(\M)}(q u_1\ldots u_n s_l) =  \prod_{0 \leq j < n}\left(1{-}\left(\frac{\mu_j}{1{-}\sum_{0 \leq i < j}\mu_i}\right)\right)\cdot\frac{\mu_l}{1{-}\sum_{0 \leq i < n}\mu_i} \\
  \end{equation*}
  By induction, one can see that for all $n \geq 1$:
  \[\prod_{0 \leq j < n}\left(1{-}\left(\frac{\mu_j}{1{-}\sum_{0 \leq i < j}\mu_i}\right)\right) = 1{-}\sum_{0 \leq i < n}\mu_i\]
  and hence
  \[ \Pr_{\mathcal{B}(\M)}(q u_1\ldots u_n s_l) = \mu_l = \Pb(q,s_l)\]

  For the other direction, we observe that given $s,q$ such that $s \in \sucs(q)$ in $\M$, there is a unique path of the form $SU^{*}S$ in $\mathcal{B}(\M)$ that starts in $s$ and ends in $q$.
  By the same reasoning as above, this path has probability $\Pb(q,s)$.
  Hence, by padding a path $\pi$ in $\M$ with the corresponding states of $U$ for every step, we get a path in $\mathcal{B}(\M)$ with the same probability.
\qed\end{proof}

The following lemma relates witnessing subsystems of $\M$ and $\mathcal{B}(\mathcal{M})$.

\begin{lemma}
	\label{lem:mwsfortrees}
	Let $\mathcal{M} = (S_{\all},s_0, \Pb)$ be an acyclic Markov chain as in Setting~\ref{setting:mdp}.
  The following two statements are equivalent:
  \begin{enumerate}
	\item $\mathcal{M}$ has a subsystem $\mathcal{M}'$ with $k$ reachable states satisfying $\Pr_{\mathcal{M}'} (\lozenge\goal)\geq \lambda$.
 \item $\mathcal{B}(\mathcal{M}) = (S_{\all} \cup U,\Pb')$ has a subsystem $\mathcal{B}'$ with reachable states $S_{\all}' \cup U'$, where $S_{\all}' \subseteq S_{\all}, U' \subseteq U$, satisfying $\Pr_{\mathcal{B}'}(\lozenge\goal) \geq \lambda$ and $|S_{\all}'| = k$.
   \end{enumerate}
\end{lemma}
\begin{proof}
	``$\implies$'': Let $\mathcal{M'} = (S_{\all}',s_0,\Pb')$ be a subsystem of $\M$ s.t. $|S_{\all}'| = k$ and satisfying $\Pr_{\mathcal{M}'} (\lozenge\goal)\geq \lambda$ and assume that all states in $S_{\all}'$ are reachable in $\M'$.

	We construct a subsystem $\mathcal{B}'$ of $\mathcal{B}(\mathcal{M})$ by taking the states $S_{\all}'$ and adding all states $U' \subseteq U$ that lie on some path $((S_{\all}')^* \,U^*)^*$ of $\mathcal{B}(\mathcal{M})$.
  If a state $u \in U'$ has a transition to a state $s \in S_{\all} \setminus S_{\all}'$, this transition is redirected to $\fail$.
	
	We get a subsystem $\mathcal{B}'$ of $\mathcal{B}(\mathcal{M})$ with states $S_{\all}' \cup U'$.
	We verify that $\Pr_{\mathcal{B}'} (\lozenge\goal)= \Pr_{\mathcal{M}'}(\lozenge\goal)$ holds, where $S' = S'_{\all} \setminus \{\goal,\fail\}$:
	\[\Pr_{\mathcal{B}'} (\lozenge\goal)= \left(\Pr_{\mathcal{B}',s_0}(\lozenge p)\right)_{p \in S'} \cdot \bb' = \left(\Pr_{\mathcal{M}',s_0}(\lozenge p)\right)_{p \in S'} \cdot \bb' = \Pr_{\mathcal{M}'}(\lozenge\goal)\]
  This implies $\Pr_{\mathcal{B}'} (\lozenge\goal)\geq \lambda$.
  The first and last equivalence use the fact that the Markov chains are acyclic. The second equivalence uses \Cref{lem:propbm} and the fact that no state in $U$ has a direct transition to $\goal$.
  Here, $\bb' = (\Pb'(q,\goal))_{q \in S'}$, as in Setting~\ref{setting:mdp}.

	``$\Longleftarrow$'':
	Let $\mathcal{B}'$ be a subsystem of $\mathcal{B}$ with states $S' \cup U'$ and $\Pr_{\mathcal{B}'} (\lozenge\goal)\geq \lambda$.
	Take the subsystem of $\mathcal{M'}$ induced by $S'$.
	By the same calculation as above we get $\Pr_{\mathcal{M}'}(\lozenge\goal) = \Pr_{\mathcal{B}'}(\lozenge\goal)$.
\qed\end{proof}

\subsubsection{Minimal witnessing subsystems for tree-shaped binary DTMCs.}

Let $\mathcal{B} = (S_{\all} \cup U,s_0,\Pb)$ be an acyclic binary Markov chain such that $\Pb(u,\goal) = 0$ for all $u \in U$ and $S_{\all} \cap U = \varnothing$.
Let $S = S_{\all} \setminus \{\goal,\fail\}$ and $\bb = (\Pb(q,\goal))_{q \in S}$.
We define $|q|$ to be the number of states reachable from $q$ and $|q|_S = |q| \cap S$.

We now give an algorithm that takes a tree-shaped binary DTMC $\mathcal{B}$, $k \in \mathbb{N}$ and $\lambda \in [0,1]$ and computes in polynomial time a subsystem $\mathcal{B}'$ with $k$ states in $S$ such that $\Pr_{\mathcal{B'}, s_0} (\lozenge\goal)\geq \lambda$, given that such a subsystem exists.

The idea is to compute a function $l_q : \{0,\ldots,|q|_S\} \to [0,1]$ for every state $q$ in $\mathcal{B}$ with the following interpretation:
$l_q(i)$ describes how much probability can be achieved in state $q$ with a subsystem that is rooted in $q$ and contains $i$ states in $S$.
If $l_{s_0}(k) \geq \lambda$, we know that there is a subsystem with $k$ states in $S$ with a probability to reach $\goal$ of at least $\lambda$.
In case that we wish to compute the subsystem, we can save the corresponding subsystem for each entry $l_{s_0}(k)$.

We compute $l$ bottom-up as follows:
first, $l_q(0) = 0$ for all states $q$ apart from $\goal$, which has $l_{\goal}(0) = 1$.
If $q \in S$ is a leaf, then $l_q(1) = \bb(q)$ and if $q$ has exactly one successor $q'$, then $l_q(i + 1) = l_{q'}(i)$, for all $0 \leq i \leq |q'|_S$.

Otherwise, suppose that $q$ has two successors $q_1,q_2$ with transition probabilities $\mu_1,\mu_2$.
If $q \in S$, we compute for $0 \leq i < |q|_S$:
\begin{align}
\begin{split}
\label{eqn:lone}
l_q(i + 1) =
\max \{ &\mu_1 \cdot l_{q_1}(j) + \mu_2 \cdot l_{q_2}(i{-}j) \mid \\
&0 \leq j \leq i, \; j \leq |q_1|_S, \; i{-}j \leq |q_2|_S\}
\end{split}
\end{align}
For $q \in U$, we do not count the state $q$ in the subsystem and hence we set (for $0 \leq i \leq |q| {-} 1$):
\begin{align}
\begin{split}
\label{eqn:ltwo}
l_q(i) = 
\max \{ &\mu_1 \cdot l_{q_1}(j) + \mu_2 \cdot l_{q_2}(i{-}j) \mid \\
&0 \leq j \leq i, \; j \leq |q_1|_S, \; i{-}j \leq |q_2|_S \}
\end{split}
\end{align}

\begin{proposition}
  \label{prop:complexity}
  Let $\M = (S_{\all} \cup U,s_0,\Pb)$ be an acyclic binary DTMC as in Setting~\ref{setting:mdp} and let $S = S_{\all} \setminus \{\goal,\fail\}$.
  The functions $l_q : \{0,\ldots,q_{|S|}\} \to [0,1]$ can be computed in time $O(|S \cup U|^3)$ for every $q \in S$.
\end{proposition}
\begin{proof}
  The functions $l_q$ are computed bottom-up by using Equations~\ref{eqn:lone}~and~\ref{eqn:ltwo}.
  Computing $l_q(i)$ requires to compute the biggest out of at most $i+1$ values, hence requiring $O(|q|_S)$ computations as $i + 1 \leq |q|_S$.
  Hence, computing the vector $l_q$, which has $|q|_S$ entries, can be done in $O\left(\left(|q|_S\right)^2\right)$ time.
  As this has to be computed for every state in $|S \cup U|$, and $|q|_S \leq |S \cup U|$, $l$ can be computed in $O(|S \cup U|\cdot|S|^2)$ time.
\qed\end{proof}

\begin{lemma}
	\label{lem:bintreesubsys}
  Let $\mathcal{B} = (S_{\all} \cup U,\Pb)$ be a tree-shaped binary DTMC as in Setting~\ref{setting:mdp} such that $\Pb(u,goal) = 0$ for all states $u \in U$ and $S_{\all} \cap U = \varnothing$.
  Then
  \begin{itemize}
	\item[(i)] If $\mathcal{B}$ has a subsystem $\mathcal{B}'$ with $k$ states in $S$ and $\Pr_{\mathcal{B}', s_0} (\lozenge\goal)\geq \lambda$, then $l_{s_0}(k) \geq \lambda$.
  \item[(ii)] If $l_{s_0}(k) = \theta$, then a subsystem $\mathcal{B}'$ of $\mathcal{B}$ with $k$ states and $\Pr_{\mathcal{B}', s_0} (\lozenge\goal) = \theta$ can be computed in polynomial time.
  \end{itemize}
\end{lemma}
\begin{proof}
	\emph{(i)}: Let $\mathcal{B}'$ be a subsystem with $k$ states in $S$ and $\Pr_{\mathcal{B}'}(\lozenge\goal) \geq \lambda$.
	As before, let $|q|_S$ denote the number of states in $S$ reachable from $q \in S \cup U$ in $\mathcal{B}'$.
	We show for all states in $\mathcal{B}'$, by induction on their height, that $l_q(|q|_S) \geq \Pr_{\mathcal{B}', q}(\lozenge \goal)$.
	\begin{enumerate}
		\item{Suppose that $q$ is a leaf.
			If $q \in U$, we have
			\[l_q(|q|_S) = l_q(0) = 0 = \Pr_{\mathcal{B}',q}(\lozenge\goal)\]
			If $q \in S$, we have 
			\[l_q(|q|_S) = l_q(1) = \Pb(q,\goal) = \Pr_{\mathcal{B}', q}(\lozenge\goal)\]
		}
		\item{Suppose that $q$ has two successors $q_1,q_2$ that satisfy the property to prove and are reached with probability $\mu_1,\mu_2$.
			
			If $q \in U$, we have $|q|_S = |q_1|_S + |q_2|_S$ and by (\ref{eqn:ltwo}):
			\begin{align*}
			l_q(|q|_S) \geq & \,\,  \mu_1 \cdot l_{q_1}(|q_1|_S) + \mu_2 \cdot l_{q_2}(|q|_S {-} |q_1|_S) \\
			= & \, \, \mu_1 \cdot l_{q_1}(|q_1|_S) + \mu_2 \cdot l_{q_2}(|q_2|_S) \\
			\geq & \,\, \mu_1 \cdot \Pr_{\mathcal{B}',q_1}(\lozenge\goal) + \mu_2 \cdot \Pr_{\mathcal{B}',q_2} (\lozenge\goal) \quad \text{(I.H.)}\\
			= & \,\, \Pr_{\mathcal{B}',q}(\lozenge\goal)
			\end{align*}
			
			If $q \in S$, we have $|q|_S = |q_1|_S + |q_2|_S + 1$ and by (\ref{eqn:lone}):
			\begin{align*}
			l_q(|q|_S) \geq & \,\,  \mu_1 \cdot l_{q_1}(|q_1|_S) + \mu_2 \cdot l_{q_2}((|q|_S {-} 1) {-} |q_1|_S) + \bb(q)\\
			= & \, \, \mu_1 \cdot l_{q_1}(|q_1|_S) + \mu_2 \cdot l_{q_2}(|q_2|_S)  + \bb(q) \\
			\geq & \,\, \mu_1 \cdot \Pr_{\mathcal{B}',q_1}(\lozenge\goal) + \mu_2 \cdot \Pr_{\mathcal{B}',q_2}(\lozenge\goal)  + \bb(q)  \quad \text{(I.H.)}\\
			= & \,\, \Pr_{\mathcal{B}'(q)}(\lozenge\goal)
			\end{align*}
		}
	\end{enumerate}
	
	As $|s_0|_S = k$ by assumption, we get $l_{s_0}(k) \geq \Pr_{\mathcal{B}', s_0}(\lozenge\goal)$.
	
	\emph{(ii)}:
	We show for every state $q$, by induction on its height, that for every $k \in \{0,\ldots,|S|\}$ and $\theta \in [0,1]$: if $l_q(k) = \theta$, then we can construct a subsystem $\mathcal{Q}$ with root $q$, $k$ states in $S$ and $\Pr_{\mathcal{Q}}(\lozenge\goal) = \theta$.
	\begin{enumerate}
		\item{Suppose that $q \in S$ is a leaf.
			$l_q(k) = \theta$ implies that $k = 1$ and $\Pb(q,\goal) = \theta$.
		}
		\item{Suppose that $q$ has two successors $q_1,q_2$ that satisfy the property to prove and are reached with probability $\mu_1,\mu_2$.
			If $q \in S$, we get:
      \begin{align*}
          l_q(k) =
          \max \{ &\mu_1 \cdot l_{q_1}(j) + \mu_2 \cdot l_{q_2}(k{-}1{-}j) \mid \\
          &0 \leq j < k, \; j \leq |q_1|_S, \; k{-}1{-}j \leq |q_2|_S\}
      \end{align*}
      As long as $k \leq |q|_S$ the set above is not empty, as we can choose $j = |q_1|_S$ which satisfies the constraints.
			Let $j^*$ be such that the above maximum is obtained, which yields
			\[l_q(k) = \mu_1 \cdot l_{q_1}(j^*) + \mu_2 \cdot l_{q_2}(k{-}1{-}j^*) = \theta\]
			By induction hypothesis we get subsystems of $q_1, q_2$ with probability at least $l_{q_1}(j^*),l_{q_2}(k{-}1{-}j^*)$ and $j^*, k{-}1{-}j^*$ states in $S$, which proves the claim.
			The case for $q \in U$ is similar.
		}
	\end{enumerate}
	As $l_{s_0}(k) \geq \lambda$ holds by assumption, we can construct a subsystem of $\mathcal{B}$ with $k$ states in $S$ and probability of reaching $\goal$ at least $\lambda$.
\qed\end{proof}

\begin{proposition}\label{thm: tree-shaped algorithm}
  Let $\M = (S_{\all},s_0,\prb)$ be a tree-shaped DTMC as in Setting~\ref{setting:mdp} and $\lambda \in [0,1]$. A minimal witnessing subsystem of $\M$ for $\Pr_{s_0}(\lozenge \goal) \geq \lambda$ can be computed in polynomial time, given that $\Pr_{s_0}(\lozenge \goal) \geq \lambda$ holds.
\end{proposition}
\begin{proof}
  We describe an algorithm that computes a minimal witnessing subsystem of the tree-shaped DTMC $\M = (S,s_0,\Pb)$ for $\Pr_{s_0}(\lozenge \goal) \geq \lambda$.
  First, we compute the binarization $\mathcal{B}(\M) = (S_{\all} \cup U,s_0,\Pb') $ of $\M$ as per \Cref{def:binarization}.
  The number of states of $\mathcal{B}(\M)$ is at most $2 \cdot |S|$, as $\M$ is tree-shaped and hence has at most as many transitions as states.

  By \Cref{lem:mwsfortrees}, every witnessing subsystem of $\mathcal{B}(\M)$ for $\Pr_{s_0}(\lozenge \goal) \geq \lambda$ with $k$ states in $S$ can be mapped to a witnessing subsystem of $\M$ for the same property with $k$ reachable states.
  We compute the function $l_{s_0}$ in polynomial time (\Cref{prop:complexity}) and choose minimal $k$ such that $l_{s_0}(k) \geq \lambda$.
  By \Cref{lem:bintreesubsys}, we can compute a witnessing subsystem of $\mathcal{B}(\M)$ for $\Pr_{s_0}(\lozenge \goal) \geq \lambda$ with $k$ states in $S$, and by the same lemma this subsystem is minimal.
\qed\end{proof}

\else
\fi

\end{document}